\documentclass{lmcs} %
\pdfoutput=1

\usepackage{lastpage}
\lmcsdoi{20}{4}{1}
\lmcsheading{}{\pageref{LastPage}}{}{}%
{Apr.~07,~2023}{Oct.~02,~2024}{}

\usepackage[utf8]{inputenc}


\usepackage{amsmath}
\usepackage{amssymb}
\usepackage{hyperref}
\usepackage[capitalise]{cleveref}
\usepackage[disable]{todonotes}

\newcommand{\R}{\mathbb{R}}

\renewcommand{\epsilon}{\varepsilon}
\newcommand{\eps}{\varepsilon}
\newcommand{\size}[1]{|#1|}

\newcommand{\N}{\mathbb{N}}

\newcommand{\nnreals}{{\mathbb R}_{\geq 0}}

\newcommand{\x}{\times}

\newcommand{\eqby}[2][=]{\stackrel{\text{{\tiny{#2}}}}{#1}}
\newcommand{\eqdef}{\eqby{def}}

\newcommand{\fractof}[1]{{\it fract}(#1)} %
\newcommand{\intof}[1]{\floor{#1}}
\newcommand{\maxfractof}[1]{{\it maxfrac}(#1)}
\newcommand{\floor}[1]{\lfloor#1\rfloor}

\newcommand{\card}[1]{|#1|}

\newcommand{\B}{\mathcal{B}}
\newcommand{\T}{\mathcal{T}}

\newcommand{\auta}{\mathcal{A}}
\newcommand{\autb}{\mathcal{B}}

\newcommand{\gsgame}{\mathcal{G}}

\newcommand{\comp}[2]{#2\circ #1}
\newcommand{\strat}{\sigma}

\newcommand{\exptime}{\EXPTIME}
\newcommand{\ptime}{\textsc{PTime}}

\newcommand{\NP}{\textsc{NP}}

\newcommand{\states}{Q}
\newcommand{\transitions}{\Delta}
\newcommand{\clocks}{C}
\newcommand{\valuations}{\nnreals^\clocks}
\newcommand{\acc}{\mathit{Acc}}
\newcommand{\priority}{\mathit{priority}}
\newcommand{\lang}{L_T}  %
\newcommand{\LTSlang}{L}  %

\newcommand{\TAtuple}{(\states,\iota,\clocks,\transitions,\Sigma,\acc)}

\newcommand{\langincl}[1][L]{\subseteq_{L}}
\newcommand{\fairsim}[1][]{\preceq}

\newcommand{\dur}[1]{\mathop{duration}(#1)}
\newcommand{\tword}[1]{\mathop{tword}(#1)}
\newcommand{\step}[2][]{\Step{#2}{}{#1}}
\newcommand{\Step}[3]{\ensuremath{\,{\stackrel{#1}{\longrightarrow}}{}^{\scriptstyle{#2}}_{\scriptstyle{#3}}}\,}

\newcommand{\tstep}[2][]{\xrightarrow[#1]{#2}}

\newcommand{\lts}{S}
\newcommand{\ltsStates}{V}
\newcommand{\ltsAlphabet}{\Sigma}
\newcommand{\ltsTransitions}{E}
\newcommand{\ltsDef}{\lts = (\ltsStates, \ltsAlphabet, \ltsTransitions)}

\newcommand{\req}{\sim}
\newcommand{\cmax}[1]{c_{#1}}

\newcommand{\EXPTIME}{{\sc ExpTime}}

\newcommand{\mypar}[1]{\paragraph*{#1.}}

\newcommand{\block}[2]{b_{#1,#2}}
\newcommand{\blimit}{\frac{1}{3}}
\newcommand{\bstart}[2]{s_{#1,#2}}
\newcommand{\bend}[2]{e_{#1,#2}}
\newcommand{\tick}[3]{f_{#1,#2,#3}}

\newcommand{\confstate}[3]{q_{#1,#2,#3}}
\newcommand{\confval}[3]{\nu_{#1,#2,#3}}

\newcommand{\envalias}[2]{\newenvironment{#1}{\begin{#2}}{\end{#2}}}
\envalias{theorem}{thm}
\envalias{corollary}{cor}
\envalias{lemma}{lem}
\envalias{proposition}{prop}
\envalias{remark}{rem}
\envalias{definition}{defi}
\crefname{defi}{Definition}{Definitions}
\crefname{cor}{Corollary}{Corollarys}
\crefname{rem}{Remark}{Remarks}
\crefname{lem}{Lemma}{Lemmas}
\crefname{thm}{Theorem}{Theorems}
\crefname{prop}{Proposition}{Propositions}

\begin{document}
	
	\title{History-deterministic Timed Automata\rsuper*}
	\titlecomment{{\lsuper*}
		This work has in parts been presented at the 
		33rd International Conference on Concurrency Theory (CONCUR'22) \cite{HLT2022} and
		at the 16th International Workshop on Reachability Problems (RP'22) \cite{BHLST2022}.
	}
	\thanks{
		This work was supported by the EU (ERC-2020-AdG 101020093); the EPSRC (EP/V025848/1,  EP/X042596/1, EP/X017796/1 and EP/X03688X/1); and the ANR (QUASY 23-CE48-0008-01).
	}%

	\author[S.~Bose]{Sougata Bose\lmcsorcid{0000-0003-3662-3915}}[a]
	\author[T.A.~Henzinger]{Thomas A.~Henzinger\lmcsorcid{0000-0002-2985-7724}}[b]
	\author[K.~Lehtinen]{Karoliina Lehtinen\lmcsorcid{0000-0003-1171-8790}}[c]
	\author[S.~Schewe]{Sven Schewe\lmcsorcid{0000-0002-9093-9518}}[a]
	\author[P.~Totzke]{Patrick Totzke\lmcsorcid{0000-0001-5274-8190}}[a]

	\address{University of Liverpool}	%
	\email{\{sougata,svens,totzke\}@liverpool.ac.uk}
	
	\address{IST Austria}
	\email{tah@ist.ac.at }
	
	\address{CNRS, Aix-Marseille University, LIS}
	\email{lehtinen@lis-lab.fr}

	\begin{abstract}
		We explore the notion of history-determinism in the context of timed automata (TA)
over infinite timed words.
History-deterministic (HD) automata are those in which nondeterminism can be
resolved on the fly, based on the run constructed thus far.
History-determinism is a robust property that admits different game-based
characterisations,
and HD specifications allow for game-based verification without an
expensive determinization step.

We show that the class of timed $\omega$-languages recognized by
HD timed automata strictly extends that of deterministic ones, and is strictly included in
those recognized by fully non-deterministic TA.

For non-deterministic timed automata %
it is known that universality is already undecidable for 
safety/reachability TA.
For history-deterministic TA with arbitrary parity acceptance, we show that
timed universality, inclusion, and synthesis all remain decidable and are
\exptime-complete.

For the subclass of TA with safety or reachability acceptance,
one can decide (in \exptime) whether such an automaton is history-deterministic,
and if so, it can be effectively determinized without introducing new automaton states.

	\end{abstract}
	
	\maketitle

	\section{Introduction}
	\label{sec:intro}

Automata offer paradigmatic formalisms both for specifying and for modelling
discrete transition systems,
\emph{i.e.}~for providing descriptive as well as executable definitions of
formal languages.
Given a finite or infinite word, an automaton specifies whether or not the
word belongs to the defined language.
Deterministic automata are executable, because the word can be processed
left-to-right, with each transition of the automaton determined by
the current input letter.
Descriptive automata allow the powerful concept of nondeterminism, which
yields more succinct or even more expressive specifications.

The notion of \emph{history-determinism} lies between determinism and
nondeterminism.
History-deterministic automata are still executable, provided the execution
engine is permitted to keep a record of all past inputs. 
Formally, a strategy~$r$ (\emph{a.k.a.}~``resolver'') is a function from
finite prefix runs to transitions that suggests for each input word $w$
a specific run $r^*(w)$ of the automaton over~$w$,
namely, the run that results from having the function $r$ determine, after
each input letter, the next transition based on the prefix of the word
processed so far.
An automaton is \emph{history-deterministic} if there exists a resolver $r$
so that for every input word~$w$, the automaton has an accepting run over
$w$ iff the specific run $r^*(w)$ is accepting.

The concept of history-determinism was first identified in~\cite{HP06},
where it was noted that for solving graph games, it is not necessary to
determinize history-deterministic specifications of $\omega$-regular
winning conditions.
For this reason, history-deterministic automata were called
``good-for-games''. %
The term ``history-determinism'' was first used by~\cite{Col09}.
The concept itself has since been referred to %
as both
``history-determinism'' and ``good-for-gameness.''
Since \cite{BL21} recently showed that, in a general context of quantitative
automata, the two notions do not always coincide
(specifically: for certain quantitative winning conditions, history-determinism
implies the ``good-for-games'' property of an automaton, but not vice versa),
we follow their more nuanced terminology and use the term ``history-determinism''
to denote the existence of a resolver and ``good-for-games'' for automata that preserve the winner of games under composition, as required for solving games without determinization.

There is also a tight link between a variant of the Church synthesis
problem, called \textit{good-enough synthesis}~\cite{AK20}, and deciding
history-determinism.
Church synthesis asks whether a system can guarantee that its interaction with
an uncontrollable environment satisfies a specification language for all possible
environment behaviours.
This model assumes that the environment is hostile and will, if possible,
sabotage the system's efforts.
This pessimistic view can be counter-productive.
In the canonical example of a coffee machine,
if the users (the environment) do not fill in the water container, the machine
will fail to produce coffee.
Church synthesis would declare the problem unrealisable: the machine
may not produce coffee for all environment behaviours.
In the good-enough synthesis problem, on the other hand, such failures are
acceptable, and we can still return an implementation that produces coffee
(satisfies the specification) whenever the environment behaves in a way that
allows the desired behaviour (fills in the water container).
Deciding the good-enough synthesis problem for a deterministic automaton is
polynomially equivalent to deciding whether a non-deterministic automaton of
the same type is history-deterministic \cite{FLW20,BL21,GJLZ21}.
The decidability and complexity
of checking history-determinism is therefore particularly interesting.

In this paper, we study, for the first time, history-determinism in the context
of \emph{timed} automata.
In a timed word, letters alternate with time delays, which are nonnegative real
numbers.
The resolver gets to look not only at all past input letters,
but also at all past time delays, to suggest the next transition.
We consider timed automata over infinite timed words with standard
$\omega$-regular acceptance conditions~\cite{AD1994}.
For the results of this paper, it does not matter whether or not the sum of all
time delays provided by an infinite input word is required to diverge.

Our results can be classified into two parts.
The first part of our results applies to all timed automata,
and sometimes more generally, to all labelled transition systems.
In this part we are concerned with solving the quintessential verification problem
for timed systems, namely \emph{timed language inclusion},
in the special case of history-deterministic (\emph{i.e.} executable) specifications.
Since universality is undecidable for general timed automata, so is the
timed language-inclusion problem for non-deterministic specifications~\cite{AD1994}.
This is the reason why much previous work in timed verification has focused on
identifying determinizable subclasses of timed automata,
such as event-clock automata~\cite{AFH99},
and on studying deterministic extensions of the timed-automaton model,
such as deterministic two-way timed automata~\cite{AlHe92}.
Determinizable specifications can be complemented, thus supporting the
\emph{complementation-based} approach to language inclusion:
in order to check if every word accepted by the implementation $A$ is also accepted
by the specification~$B$, first determinize and complement~$B$, and then check the
intersection with $A$ for emptiness.
We show that the history-determinism of specifications suffices for deciding timed
language inclusion, which demonstrates that determinizability is not required.
More precisely, we prove that if $A$ is a timed automaton and $B$ is a
history-deterministic timed automaton, it can be decided in \exptime\ if every timed
word accepted by $A$ is also accepted by~$B$ (\cref{cor:inclusion}).

In contrast to the traditional complementation-based approach to language inclusion,
the history-deterministic approach is \emph{game-based}.
Like the complementation-based approach, the game-based approach is best formulated
in the generic setting of labelled transition systems with acceptance conditions,
so-called \emph{fair LTS}.
The acceptance condition of a fair LTS declares a subset of the infinite runs of the
LTS to be fair
(a special case is \emph{safety} acceptance, which declares all infinite runs to be
fair). 
Given two fair LTS $A$ and~$B$, the language of $A$ is included in the language of $B$
if for every fair run of $A$ there is a fair run of $B$ over the same (infinite) word.
A sufficient condition for the language inclusion between $A$ and $B$ is the existence
of a fair simulation relation between the states of $A$ and the states of~$B$,
or equivalently, the existence of a winning strategy for player $p_B$ in the following
2-player \emph{fair simulation game}:
(i)~every transition chosen by player $p_A$ on the state-transition graph $A$
can be matched by a transition chosen by player $p_B$ on the state-transition graph $B$
with the same label (letter or time delay), and
(ii)~if the infinite sequence of transitions chosen by $p_A$ produce a fair run of~$A$,
then the matching transitions chosen by $p_B$ produce a fair run
of~$B$~\cite{HKR1997}.
Solving the fair simulation game is often simpler than checking language inclusion;
it may be polynomial where language inclusion is not
(\emph{e.g.} in the case of finite safety or B\"uchi automata),
or decidable where language inclusion is not
(\emph{e.g.} in the case of timed safety or B\"uchi automata~\cite{TAKB96}).

We show that for all fair LTS $A$ and all history-deterministic fair LTS~$B$,
the condition that the language of $A$ is included in the language of $B$ is
equivalent to the condition that $A$ is fairly simulated by~$B$.
This observation reduces the language inclusion problem for history-deterministic
specifications to the problem of solving a fair simulation game between
implementation and specification.
The solution of fair simulation games depends on the complexity of the acceptance
conditions of $A$ and~$B$, but is often simpler than the complementation of~$B$,
and fair simulation games can be solvable even in the case of specifications that
cannot be complemented.
In \cref{sec:expressivity:d-hd} we show the existence of such a timed language.
The game-based approach to checking language inclusion,
which requires history-determinism,
is therefore more general, and often more efficient,
than the traditional complementation-based approach to checking language inclusion,
which usually requires full determinization.
Indeed, history-determinism is exactly the condition that allows the game-based
approach to language inclusion:
for a given fair LTS~$B$, if it is the case that $B$ can fairly simulate all fair
LTS $A$ whose language is included in the language of~$B$,
then $B$ must be history-deterministic (\cref{thm:triangle}).

More generally, turn-based timed games for which the winning
condition is defined by a history-deterministic timed automaton are no harder to solve than those with deterministic winning conditions:
the winner of such a timed game can be determined on the product of the (timed)
arena with the automaton specifying the winning condition.
We conjecture that this is the case
also for the concurrent timed games of~\cite{DFHM03}
(cf.~\cref{sec_conclusion}).
Timed games have also been defined for the synthesis of timed systems from timed I/O
specifications.
Again, we show that the synthesis game of~\cite{DM2002} can be solved not only for
I/O specifications that are given by deterministic timed automata,
but more generally, for those given by history-deterministic
timed automata (\cref{thm:synthesis}).

The second part of our results investigates the problem of deciding
history-determinism for timed automata and the determinizability of
history-deterministic timed automata.
In this part, we have only partial results,
namely results for timed safety and reachability automata.
Timed safety automata, in particular, constitute an important class of
specifications, as many interesting timed and untimed properties
can be specified by timed safety automata if time is required to
diverge~\cite{HNSY1992,HKW95}.
We prove that for timed safety automata and timed reachability automata,
it can be decided in \exptime\ if a given timed automaton is
history-deterministic (\cref{thm:decidable}).
Checking history-determinism remains open for more general classes of
timed automata, such as timed B\"uchi and coB\"uchi automata.
We also show that every history-deterministic timed safety or reachability automaton can be determinized without increasing the number of automaton states, but with an exponential
increase in the number of transitions or length of guards (Theorem 4.5).
Since the question of determinizability is undecidable for non-deterministic
timed reachability automata~\cite{F06},
it follows from our result, that checking if a given
non-deterministic timed reachability automaton has an equivalent HD timed reachability automaton
is also undecidable.
Finally, we show that if a timed safety or reachability automaton is
good-for-games (in the sense explained earlier),
then the automaton must be history-deterministic (\cref{lem:gfg2hd}).
This implication is open for more general classes of timed automata.

All our results hold regardless of whether one assumes weak or strong progression of time (zero-delays are allowed, resp.~forbidden) and also whether time must diverge, i.e., infinite words of finite duration (Zeno words) are considered or not.

\medskip
\mypar{Related Work}
The notion of history-determinism was introduced independently, with slightly
different definitions, by Henzinger and Piterman~\cite{HP06} for solving games without determinization, by Colcombet~\cite{Col09} for  cost-functions, and by Kupferman, Safra, and Vardi~\cite{KSV06} for recognizing derived tree languages of word automata.
Initially, history-determinism was mostly studied in the $\omega$-regular setting,
where these different definitions all coincide~\cite{BL19}.
For some coB\"uchi-recognizable languages, history-deterministic automata can be
exponentially more succinct than any equivalent deterministic automaton~\cite{KS15},
and for B\"uchi and coB\"uchi automata, history-determinism is decidable in
polynomial time~\cite{BK18,KS15}.
For transition-based history-deterministic automata, minimisation is
\ptime~\cite{RK19}, while for state-based ones, it is \NP-complete~\cite{Sch20}.
Recently, the notion has been extended to richer automata models, such as
pushdown automata~\cite{LZ22,GJLZ21} and quantitative automata~\cite{BL21,BL22},
where deterministic and non-deterministic models have different expressivity,
and therefore, allowing a little bit of non-determinism can, in addition to
succinctness, also provide more expressivity.

\medskip
\mypar{Paper Structure}
After defining preliminary notions we proceed to introduce history-determinism,
and show a new, fair-simulation-based characterisation in \cref{sec:gfg}.
In \cref{sec:expressivity} we consider the expressivity of history-deterministic(HD) timed automata(TA) with different Parity acceptance conditions.
We show that history-deterministic TA with safety or reachability acceptance are determinizable.
\Cref{sec:games}
considers questions concerning timed games, timed synthesis, and timed language inclusion
and shows that history-determinism coincides with good-for-gameness for reachability and safety TA.
In \cref{sec:decision} we derive that one can decide, in \EXPTIME, if a given safety or reachability TA is history-deterministic.
In \cref{sec:synthesis} we study synthesis games with winning conditions given by history-deterministic TA 
and show that solving the corresponding synthesis game remains in \EXPTIME, as for deterministic TA.
We conclude with a summary and some open problems and conjectures.

	\section{Preliminaries}
	\label{sec:definitions}
	\mypar{Numbers, Words}
Let $\N$ and $\nnreals$ denote the nonnegative integers and reals, respectively. 
For $c\in\nnreals$ we write
$\intof{c}$ %
for its integer
and $\fractof{c} \eqdef c-\lfloor{c}\rfloor$
for its fractional part.

An alphabet~$\Sigma$ is a nonempty set of letters. $\Sigma_{\varepsilon}$ denotes $\Sigma\cup \{\varepsilon\}$.  $\Sigma^*$ and $\Sigma^\omega$ denote the sets of finite and infinite words over $\Sigma$, respectively and $\Sigma^\infty=\Sigma^*\cup \Sigma^\omega$ denotes their union.
The  empty word is denoted by $\epsilon$, the length of a finite word~$v$ is denoted by $\size{v}$, and the $n$-th letter of a finite or infinite word is denoted by $w[n]$ (starting with $n = 0$).

\medskip
\mypar{Labelled Transition Systems, Languages, Fair Simulation}
A \emph{labelled transition system} (LTS)
is a graph 
$S=(V,\ltsAlphabet,\ltsTransitions)$
with set $V$ of states %
and edges $\ltsTransitions\;\subseteq\; V \x \ltsAlphabet \x V$, 
labelled by alphabet $\ltsAlphabet$.
It is \emph{deterministic} if for all $(s,a)\in V\x\ltsAlphabet$ there is at most one $s'$ with $s\step{a}s'$,
and \emph{complete} if for all $(s,a)\in V\x\ltsAlphabet$ there is at least one $s'$ with $s\step{a}s'$.
We henceforth consider only complete LTSs.
Together with an \emph{acceptance condition}
$\acc\subseteq \ltsTransitions^\omega$ this can be used to define languages over $\ltsAlphabet$ as usual:
a word $w=l_0l_1\ldots\in\ltsAlphabet^\omega$ is accepted from $s_0$ if
there is a path (also \emph{run})
$\rho=
s_0
\step{l_1}
s_1
\step{l_2}
s_2
\ldots
$
that is accepting, i.e., in $\acc$.
The \emph{language} $\LTSlang(s_0)\subseteq \ltsAlphabet^\omega$ of an initial state $s_0\in V$ 
consists of all words for which there exists an accepting run from $s_0$.
We will write $s\langincl s'$ to denote language inclusion, meaning $\LTSlang(s)\subseteq \LTSlang(s')$.
The acceptance condition $\acc$ can be given by a parity condition 
but does not have to be.
We consider in this paper especially reachability (does the run visit a state in a given target set $T\subseteq V$?)
and safety conditions (does the run always stay in a ``safe'' region $F\subseteq V$?).
An LTS together with an accepting condition is referred to as \emph{fair} LTS \cite{HKR1997}.

\emph{Fair} simulations 
\cite{HKR1997}
are characterised by simulation games on (a pair of) fair LTSs in which  Player 1 stepwise produces a path from $s$, and Player 2 stepwise produces an equally labelled path from $s'$. Player 2 wins if she produces an accepting run
whenever Player 1 does.
That is, $s$ is fairly simulated by $s'$ (write $s\fairsim s'$) iff Player 2 has a strategy in the simulation game
so that, whenever the run produced by Player 1 is accepting then so is the run produced by Player 2 in response.
Fair simulation $s\fairsim s'$ implies language inclusion $\LTSlang(s)\subseteq \LTSlang(s')$ but not vice versa.

\medskip
\mypar{Timed Alphabets, Words, and LTSs}
For any alphabet $\Sigma$ let $\Sigma_T$ denote the timed alphabet $\{(a,t)|a\in \Sigma, t\in \R_\ge0\}$. 
A timed word is a finite or infinite word $w\in (\Sigma_T)^\infty$ consisting of letters in $\Sigma$ paired with distinct non-negative non-decreasing real-valued timestamps. %
We will also write $d_0a_0d_1a_1...$ to denote a timed word $(a_i,t_i)\in\Sigma_T^\infty$ where $t_0=d_0$ and $t_{i+1}=t_i+d_{i+1}$.
Conversely, the duration and the timed word
of any sequence in $(\Sigma\cup\R)^\infty$
is given inductively as follows.
For any
$d\in\nnreals$,
${\tau\in\Sigma}$,
$\alpha\in(\Sigma\cup\R)^*$,
and
$\beta\in(\Sigma\cup\R)^\infty$
let
$\dur{\tau}\eqdef 0$;
$\dur{d}\eqdef d$;
$\dur{\alpha\beta}=\dur{\alpha}+\dur{\beta}$;
$\tword{{\eps}}=\tword{d}\eqdef\eps$;
$\tword{{\alpha d}}\eqdef\tword{\alpha}$; and %
$\tword{{\alpha \tau}}\eqdef\tword{\alpha}(\tau,\dur{\alpha})$.
An infinite timed word of finite duration is called a \emph{Zeno} word. 

A \emph{timed} LTS is one with edge labels in $\Sigma\uplus\nnreals$,
so that edges labelled by $\nnreals$ (modelling the passing of time) satisfy the following conditions
for all $\alpha,\beta,\gamma\in V$ and $d,d'\in \R_{\ge 0}$.
\begin{enumerate}
    \item (Zero-delay): $\alpha\tstep{0}\alpha$,
    \item (Determinism): If
        $
        \alpha\tstep{d}\beta
        \land
        \alpha\tstep{d}\gamma$ then $\beta=\gamma$,
    \item (Additivity): $\alpha\tstep{d}\beta\tstep{d'}\gamma$ then $\alpha\tstep{d+d'}\gamma$.
\end{enumerate}
The {timed language} $\lang(s)\subseteq \Sigma_T^\omega$ of a state $s$
consists of all the timed words read along accepting runs $\lang(s)\eqdef \tword{\LTSlang(s)}$. We write $L(\lts)$ for the timed language of the initial state of the LTS $\lts$.

\medskip
\mypar{Timed Automata}
Timed automata are finite-state automata equipped with
finitely many real-valued variables called \emph{clocks}, whose transitions are guarded by constraints on clocks.
Constraints on clocks $\clocks=\{x,y,\ldots\}$ are (in)equalities $x-y \triangleleft n$ where $x\in \clocks$, $y\in \clocks\cup \{0\}$, $n\in \mathbb{N}$ and $\triangleleft \in \{\leq, <\}$. Let $\B(\clocks)$ denote the set of Boolean combinations of clock constraints, called \emph{guards}.
A clock \emph{valuation} $\nu\in \R^\clocks$ assigns a value $\nu(x)$ to each clock $x\in \clocks$. We write $\nu \models g$ if $\nu$ satisfies the guard $g$. 
A timed automaton (TA) $\T=\TAtuple$ is given by 
\begin{itemize}
\item $Q$ a finite set of states including an initial state $\iota$;
\item $\Sigma$ an input alphabet;
\item $\clocks$ a finite set of clocks;
\item $\transitions\subseteq Q\times \B(\clocks)\times \Sigma\times 2^\clocks \times Q$  a set of transitions; each transition is associated with a guard, a letter, and a set of clocks to reset. 
    A transition that reads letter $a\in\Sigma$ will be called an $a$-transition.
    We assume that %
    for all $(s,\nu,a)\in Q\x\valuations\x\Sigma$ there is at least one
    transition $(s,g,a,r,s')\in\transitions$
    so that $\nu$ satisfies $g$.
\item $\acc\subseteq \Delta^\omega$ an acceptance condition. 
\end{itemize}
Timed automata induce timed LTSs, and can thus be used to define timed languages, as follows.
A \emph{configuration} is a pair consisting of
a control state and a clock valuation.
These can evolve in two ways, as follows.
For all configurations $(s,\nu)\in \states\x\valuations$,
\begin{itemize}
    \item 
        there is a \emph{delay} step 
        $(s,\nu)\step{d}(s,\nu+d)$
        for every $d\ge0$, which increments all clocks by $d$.
    \item there is a \emph{discrete} step
        $(s,\nu)\step{\tau}(s',\nu')$ 
        if
        $\tau=(s,g,a,r,s')\in\transitions$
        is a transition so that $\nu$ satisfies $g$ and $\nu'=\nu[r\rightarrow 0]$, that is, it maps $r$ to $0$ and agrees with $\nu$ on all other values.
\end{itemize}
Naturally, every delay $d$ yields a unique successor configuration
and so, for any two $d,d'\ge 0$ and valuations $\nu,\nu'$,
$\nu\step{d}\step{d'}\nu'\iff\nu\step{d+d'}\nu'$.
Hence, a TA indeed induces a timed LTS.

Discrete steps, however, are a source of nondeterminism: 
a configuration may have several $a$-successors
induced by different transitions whose guards are satisfied.
$\T$ is \emph{deterministic} if its induced LTS is deterministic, which is the case iff for every state $s$, all transitions from $s$ have mutually exclusive guards.

A path
$\rho=
(s_0,\nu_0)
\step{l_1}
(s_1,\nu_1)
\step{l_2}
(s_2,\nu_2)
\ldots
$
is called \emph{reduced} if it does not contain consecutive delay steps.
It is a \emph{run on} timed word
$w
\in (\Sigma_T)^\infty$
if
$\tword{l_1l_2\ldots}=w$.
The acceptance condition is lifted to the LTS as expected. Namely,
a run is \emph{accepting}
if $\rho\in\acc$.
This way,
the \emph{language} $\lang(s,\nu)\subseteq \Sigma_T^\omega$ of a configuration $(s,\nu)$
consists of all timed words for which there exists an accepting run from $(s,\nu)$.
The language of $\T$ is $\lang(\T)\eqdef\lang((\iota,0))$, the languages of the initial configuration with state $\iota$ and all clocks set to zero.

\medskip
\mypar{Region Abstraction}
The following is the standard definition of regions for timed automata (cf.~\cite{AD1994}, def.~4.3).
Let $\T=\TAtuple$ be a timed automaton and
for any clock $x\in\clocks$ let $\cmax{x}$ denote the largest constant in any clock constraint involving $x$.
Two valuations $\nu,\nu'\in \nnreals^C$ are \emph{(region) equivalent} (write $\nu\req\nu'$) if all of the following hold.
\begin{enumerate}
\item For all $x\in\clocks$ either $ \floor{\nu(x)}=\floor{\nu'(x)} $
    or both $\nu(x)$ and $\nu'(x)$ are greater than $\cmax{x}$.
\item For all $x,y\in\clocks$ with $\nu(x)\le \cmax{x}$ and $\nu(y)\le \cmax{y}$,
    $\fractof{\nu(x)}\le \fractof{\nu(y)}$
    iff
    $\fractof{\nu'(x)}\le \fractof{\nu'(y)}$.
\item For all $x\in\clocks$ with $\nu(x)\le \cmax{x}$, $\fractof{\nu(x)}=0$ iff $\fractof{\nu'(x)}=0$.
\end{enumerate}
It follows that there are only finitely many equivalence classes w.r.t.~$\req$, called regions, for any given TA. Two configurations $(s,\nu)$ and $(s',\nu')$ are (region) equivalent,
write $(s,\nu)\req(s',\nu')$, if $s=s'$ and $\nu\req\nu'$.
Two runs are (region) equivalent if they have the same length and stepwise visit region equivalent configurations.
Let $\maxfractof{\nu} = \max\{\fractof{\nu(x)} \mid x\in\clocks\}$ denote the maximal fractional value of any clock in configuration $\nu$.
We will make use of the following two properties.

\begin{proposition}\
    \label{lem:paths}
\begin{enumerate}
    \item 
    \label{lem:paths-short-delay}
For any valuation $\nu$ and $d\le 1- \maxfractof{\nu}$
we have $\nu\req\nu+d$.

    \item 
    \label{lem:paths-short-path}
    Suppose that
    $(p,\nu)\req(p',\nu')$
    and let
    $\rho\in (\transitions\cup\nnreals)^*$
    satisfy
    $\dur{\rho} < 1-\maxfractof{\nu}$,
    $\dur{\rho} < 1-\maxfractof{\nu'}$
    and
    $(p,\nu)\step{\rho}(q,\mu)$.
    
    Then 
    $(p',\nu')\step{\rho}(q',\mu')$
    for some 
    $(q',\mu')\req(q,\mu)$.
\end{enumerate}
\end{proposition}
\begin{proof}[Proof sketch]
    Part 1 is immediate from the definition of regions.

    Part 2 can be shown by induction on the length of $\rho$
    using the facts that region-equivalent configurations enable the same discrete transitions
    and that
    any delay decreases the duration of the remaining path by the same amount it increases clocks.
    \qedhere

\end{proof}

	\section{History-determinism}
	\label{sec:gfg}
	Informally, an  automaton or LTS is history-deterministic if the nondeterminism can be resolved on-the-fly, based only on the history of the word and run so far. We give two equivalent definitions, each being more convenient than the other for some technical developments. 
The first one extends the original definition of history-deterministic  automata (defined as good-for-games automata in \cite{HP06}) to all fair LTSs. 

\begin{definition}[History-determinism]
    A fair LTS $\ltsDef$ is \emph{history-deterministic} (from initial state $s_0\in V$) if there is a \emph{resolver} $r:E^*\x\Sigma\to E$ that maps every finite run and letter $a\in \Sigma$ to an $a$-labelled transition
    such that, for all words $w=a_0a_1\dots \in \lang(s_0)$
    the run $\rho$ defined inductively for $i>0$ by $\rho_{i+1} \eqdef \rho_ir(\rho_i, a_{i+1})$,
    is an accepting run on $w$ from $s_0$.
\end{definition}

Equivalently (from~\cite{BL19} for $\omega$-regular automata), a resolver corresponds exactly to a winning strategy for Player 2 in the following \textit{letter game}.
\begin{definition}[Letter game]
The letter game on a fair LTS $\ltsDef$ with initial state $s_0\in V$ is played between Players 1 and 2. At turn $i$:
\begin{itemize}
\item Player 1 chooses a letter $a_i\in \Sigma$.
\item Player 2 chooses an $a_i$ labelled edge $\tau_i \in E$.
\end{itemize}
A play is a pair $(w,\rho)$ where $w=a_0a_1\ldots$ is an infinite word and $\rho=\tau_0\tau_1...$ is a run on $w$. A play is winning for Player 2 if either $w\notin \lang(s_0)$ or $\rho$ is an accepting run on $w$ from $s_0$.
\end{definition}
In these and other games we consider, strategies for both players are defined as usual, associating finite histories (runs) to valid player choices.
Now winning strategies for Player~2 in the letter game
exactly correspond to resolvers for $\lts$ and vice-versa.
\begin{proposition}
Player 2 wins the letter game on a fair LTS $S$ if and only if $S$ is history-deterministic.
\end{proposition}

While history-determinism is known to relate to fair simulation, in the sense that history-deterministic automata simulate deterministic ones for the same language~\cite{HP06}, their relation has so far not been studied in more details. Below we show that history-determinism can equivalently be characterised in terms of fair simulation.

\begin{theorem}\label{thm:triangle}
For every fair LTS $\lts$ %
and initial state $q$
the following are equivalent:
\begin{enumerate}
\item \label{triangle-hd} $\lts$ is history-deterministic.
\item \label{triangle-sim}For all complete fair LTS $\lts'$ with initial state $q'$, %
    $q' \langincl q$ if and only if $q' \fairsim q$. 
\end{enumerate}
\end{theorem}

\begin{proof}[Proof (\textbf{\ref{triangle-hd})$\implies$(\ref{triangle-sim})}]
Fair simulation $q \fairsim q'$ trivially implies $q \langincl q'$ by definition.

For the other implication, assume that $q \langincl q'$.
By assumption (\ref{triangle-hd}) there exists a resolver, i.e.~a winning strategy in the letter game.
    Player~2 can win the fair simulation game by ignoring her opponent's configuration and moving according to this resolver. By the completeness assumption on $\lts'$,
    Player~1 can never propose a letter for which there is no successor in $\lts'$. So each player produces an infinite run on the same word $w$ and the run produced by Player~2 is the same as that produced by the resolver in $\lts'$. 
    If $w\in\lang(q)$ then it is in $\lang(q')$ and Player~2's run accepts.
    If $w\notin\lang(q)$ then Player~2 wins due to the fairness condition.
    In both cases she wins the fair simulation game and therefore $q \fairsim q'$.

    \textbf{(\ref{triangle-sim})$\implies$(\ref{triangle-hd})}
    If condition (\ref{triangle-sim}) holds for all complete fair LTSs then $q$ can fairly simulate the one
    consisting of a single state with self-loops for all transitions of $\lts$ whose acceptance condition contains exactly all accepting runs from $q$. Then the strategy for Player~2 in the fair simulation game
    can be used as a strategy in the letter game.
\end{proof}

	\section{Expressivity}
	\label{sec:expressivity}
	We consider the expressivity of history-deterministic TA in comparison to deterministic and fully non-deterministic variants, for different accepting conditions.
In particular, we prove (in \cref{sec:expressivity:safety-reach}) 
that for safety and reachability acceptance, HD TA can be determinized,
whereas for coB\"uchi and above they cannot (\cref{sec:expressivity:d-hd}).
Moreover, we show (in \cref{sec:expressivity:hd-nd}) that Parity HD TA are strictly less expressive than fully non-deterministic TA, even those with only a reachability acceptance.

\subsection{Safety and Reachability HD TA are determinizable}
\label{sec:expressivity:safety-reach}
We start by showing that history-deterministic timed automata with safety acceptance are determinizable. To do so, we show (in \cref{lem:region-resolver}) that these automata have simple resolvers, which only depend on the equivalence class of the current clock configuration with respect to the region abstraction. That is to say, the resolver only needs to know the integer part of clock values (up to the maximal value that appears in clock constraints) and the ordering of their fractional parts. We can then use such a simple resolver to determinize the automaton by adding guards that restrict transitions so that the automaton can only take one transition per region, as dictated by the resolver.

Before proving this formally, we first consider the reachability timed automaton depicted in \cref{fig:r-hd-ex}.
The accepted language consists of timed words that have a prefix of $b$ and $c$ events, followed by an infinite sequence of $a$ events
such that
within one unit of time 
preceding the second $a$, 
there is either a $b$ or a $c$ event, 
but not both.
The automata is history-deterministic: when reading the first $a$ event, the resolver goes to $q_2$ if the previous letter was a $b$, and to state $q_3$ otherwise.
Notice that this decision can be made only based on the current valuations  and does not require the entire history, because the values of clocks satisfy $x\le y$ iff the last letter was $b$.
The automaton can therefore be determinized by adding a guard $x-y\le 0$ to the blue edge (to $q_2$) and its negation to the red edge (to $q_3$).

\begin{figure}[t]
	\centering
	\includegraphics{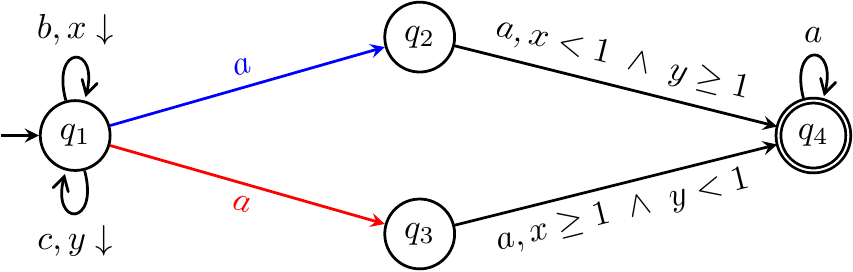}
	
	\caption{A history-deterministic timed reachability automaton where a resolver must distinguish between valuations
            where $x>y$ (go to $q_3$) and those where $x\le y$ (go to $q_2$).
        }
	\label{fig:r-hd-ex}
\end{figure}

\begin{definition}[Run-trees]
    A \emph{run-tree} on a timed word $u=(a_0,t_0)(a_1,t_1)\ldots$ from TA configuration $(s_0,\nu_0)$ is a tree
    where nodes are labelled by configurations,
    and edges by transitions such that
    \begin{enumerate}
        \item 
    The labels along every branch form a run on $u$ from $(s_0,\nu_0)$.
\item It is complete wrt.~discrete steps:
    suppose %
    the path leading towards some node is labelled by a run $\rho$
    which reads $\tword{\rho}=(a_0,t_0)\ldots(a_i,t_i)$,
    ends in a configuration $(s,\nu)$,
    and has $\dur{\rho}=t_{i+1}$. %
    Then 
    for 
    every transition $\tau=(s,g,a_{i+1},r,s')\in \transitions$
    with $\nu\models g$
    and so that $(s,\nu)\step{\tau}(s',\nu')$,
    there is a $\tau$-labelled edge to a new node labelled by $(s',\nu')$.
    \end{enumerate}
    A run-tree is \emph{reduced} if all its branches are. That is, there are no consecutive delay steps.
\end{definition}
Notice that for every initial configuration and timed word, there is a unique reduced run-tree,
all of whose branches are runs on the word (since we have no deadlocks),
and vice versa, all reduced runs on the word appear as branches on the run-tree.

We extend the region equivalence from configurations to run-trees in the natural fashion:
two run-trees are equivalent if they are isomorphic and all corresponding configurations are equivalent.
That is, they can differ only in fractional clock values and the duration of delays.

The following is our key technical lemma.

\begin{restatable}{lemma}{LemRuntreeEq}
\label{lem:runtree-eq}
Consider two region equivalent configurations
$(s,\nu)\req(s',\nu')$.

For every timed word $u$
there is a timed word $u'$
so that the reduced run-tree on $u$
from $(s,\nu)$
is equivalent to 
the reduced run-tree on $u'$ from
$(s',\nu')$.
\end{restatable}

\begin{proof}
    It suffices to show that for some (not necessarily reduced) run-tree on $u$ from $(s,\nu)$
    there exists some equivalent run-tree from $(s',\nu')$ as this implies the claim by collapsing all consecutive delay steps and thus producing the reduced tree on both sides.

    We proceed by stepwise uncovering the run-tree from $(s,\nu)$
    for ever longer prefixes of $u$
    and constructing a corresponding equivalent run-tree from $(s',\nu')$.
    The intermediate finite trees we build have the property that all branches have the same duration.
    In each round we extend all current leafs, in both trees, either by
    \begin{enumerate}
        \item all possible non-deterministic successors (for the letter prescribed by the word $u$),
    in case the duration of the branch is already equal to the next time-stamp in $u$, or
\item one successor configuration due to a delay, which must be \emph{the same on all leafs}.
    \end{enumerate}
    For the second case, the delays used to extend the two trees need not be the same
    because we only want to preserve region equivalence.
    Also, the delay chosen for the tree rooted in $(s,\nu)$ need not follow the
    timestamps in $u$ but can be shorter, meaning the run-tree may not be reduced.
    The difficulty lies in systematically choosing the delays to ensure that the two trees remain equivalent, and secondly, that in the limit this procedure generates a run-tree on the whole word $u$ from $(s,\nu)$.
    Together this implies the existence of a corresponding word $u'$ and a run-tree from $(s',\nu')$.
    
    \medskip
    \textbf{Invariant.}
    We propose a stronger invariant, namely that the relative orderings of the fractional values
    \emph{in all leafs are the same on both sides}.
    To be precise, 
    let's reinterpret a clock valuation as a function $\nu:\clocks\x\N\to\{\bot\}\cup[0,1)$,
    that assigns to every clock and possible integer value either a fractional value between $0$ and $1$,
    or $\bot$ (indicating that the given clock does not have the given integer value). 
    This way for every clock $x$ there is exactly one $n\in\N$ with $\nu(x,n)\neq\bot$
    and the image $\nu(\clocks\x\N)$ has at most $\card{\clocks}+1$ different elements.
    For any ordered set $F=\{\bot<f_1<f_2<\cdots<f_\ell\}\supseteq\nu(\clocks\x\N)$ of fractional values, we can thus represent $\nu$
    as a function $\hat{\nu}:\clocks\x\N\to \{\bot,1,\ldots, \ell\}$ that, instead of exact fractional clock values only yields their index in $F$ (and maps $\bot\mapsto\bot$).

    Consider some run-tree with leafs
    $(q_1,\nu_1)(q_2,\nu_2)\cdots(q_\ell,\nu_\ell)$
    with combined fractional values $F=\bigcup_{j=1}^\ell\nu_j(\clocks\x\N)$,
    and an equivalent run-tree with leafs
    $(q'_1,\nu'_1)(q'_2,\nu'_2)\cdots(q'_\ell\nu'_\ell)$
    with combined fractional values $F'=\bigcup_{j=1}^\ell\nu_j'(\clocks\x\N)$.
    The two trees are \emph{aligned} if
    for all $1\le j \le \ell$, $\hat{\nu}_j = \hat{\nu}'_j$.
    Notice that this still allows the two trees to differ on their exact fractional values
    but now they must agree on the relative order of all contained clocks on leafs, and in particular which ones are maximal and therefore the closest to the next larger integer.
    We will always select a delay of $1-\max\{F\}$ and $1-\max\{F'\}$, respectively, in step 2 above.

    To show the claim we produce the required run-trees starting in $(s,\nu)\req(s',\nu')$.
    These are in particular two aligned run-trees on the empty word.
    
    Assume two aligned trees as above,
    where leafs have fractional values
    $F=\{\bot<f_1<f_2<\cdots<f_m\}$ and 
    $F'=\{\bot<f'_1<\cdots<f'_m\}$, respectively,
    and assume 
    that the tree rooted in $(s,\nu)$ reads a strict prefix $(a_0,t_0),\ldots (a_i,t_i)$ of $u$.
    
    \smallskip
    \emph{Case 1: } the duration of all branches in the first tree equals $t_{i+1}$, the timestamp of the next symbol in $u$.
    Then we extend each leaf in both trees by all possible $a_{i+1}$-successors.
    This will produce two aligned trees because
    each leaf configuration in one must be region equivalent to the corresponding configuration in the other,
    and therefore satisfies the same guards, enabling the same $a_{i+1}$-transitions leading to equivalent successors. Note also that all branches in each tree still have the same duration, as no delay step was taken.

    \smallskip
    \emph{Case 2: } the duration of all branches in the first tree is strictly less than $t_{i+1}$.
    Then we extend all leafs in the tree from $(s,\nu)$ by a delay of duration $d=1-f_m$
    and all leafs in the other tree by a delay of duration $d'=1-f'_m$.
    Naturally, this produces exactly one successor for each former leaf.
    The sets of new fractional values on leafs are
    $\bigcup_{j=1}^\ell(\nu_j+d)(\clocks\x\N) = \{\bot<0<f_1+d<\cdots<f_{m-1}+d\}$,
    and for any former leaf $(q,\mu)= (q_j,\nu_j)$ for $0\le j\le \ell$, extended by a delay $(q,\mu)\step{d}(q,\mu+d)$,
    we have 
    \begin{equation}
        \label{eq:runtree-e1}
        \hat{\mu}(x,n-1) = m \iff \widehat{(\mu+d)}(x,n) = 0
    \end{equation}
    and 
    \begin{equation}
        \label{eq:runtree-e2}
        \hat{\mu}(x,n) = j < m \iff \widehat{(\mu+d)}(x,n) = j+1\le m
    \end{equation}
    Analogous equivalences hold for the corresponding step
    $(q,\mu')\step{d'}(q,\mu'+d')$ on the other tree.
    Notice that the two cases above are exhaustive as 
    again, for all $x\in\clocks$ there is exactly one $n\in\N$ with
    $\mu(x,n)\neq\bot$.
    We aim to show that $\widehat{(\mu+d)}=\widehat{(\mu'+d')}$.
    Consider any $x\in\clocks$ and $n\in\N$. We have that
    \begin{align*}
    \widehat{(\mu+d)}(x,n) = m  %
        &\stackrel{(\ref{eq:runtree-e1})}{\iff}
        \hat{\mu}(x,n+1)=0\\
        &\stackrel{(IH)}{\iff}
        \hat{\mu'}(x,n+1)=0\\
        &\stackrel{(\ref{eq:runtree-e1})}{\iff}
    \widehat{(\mu'+d')}(x,n) = m 
    \end{align*}
    and
    \begin{align*}
    \widehat{(\mu+d)}(x,n)=j < m%
        &\stackrel{(\ref{eq:runtree-e2})}{\iff}
        \hat{\mu}(x,n)=j+1\\
        &\stackrel{(IH)}{\iff}
        \hat{\mu'}(x,n)=j+1\\
        &\stackrel{(\ref{eq:runtree-e2})}{\iff}
        \widehat{(\mu'+d')}(x,n) = j<m
    \end{align*}
    It follows that $\widehat{(\mu+d)}=\widehat{(\mu'+d')}$ which means that the two trees are again aligned, as required.

    To see why this procedure produces a run-tree on $u$ (and an equivalent run-tree on some word $u'$),
    observe that 
    there can be at most $\card{F}+1$ many consecutive delay extensions
    according to step 2) before all integral clock values are strictly increased.
\end{proof}

We are now ready to show that history-deterministic TA with safety acceptance have simple resolvers based on the region abstraction.

We call a resolver $r$
\emph{region-based} if it
bases its decision only on the current letter and region.
That is, 
if for any letter $a\in\Sigma$
and any two finite runs $(\iota,0)\step{\rho}(s,\nu)$ and $(\iota,0)\step{\rho'}(s',\nu')$
consistent with $r$ and so that $(s,\nu)\req(s',\nu')$,
it holds that $r(\rho,a)=r(\rho',a)$.
 
 \begin{lemma}
     \label{lem:region-resolver}
     Every history-deterministic TA with safety acceptance has a region-based resolver.
     The same is true for history-deterministic TA over finite words.
\end{lemma}
\begin{proof}
Let $r$ be a resolver for a history-deterministic safety TA $\T$. %

We now build a resolver that only depends on the region of the current configuration. To do so, we choose a representative configuration within each region, which will determine the choice of the resolver for the whole region:
For every region $R\in [\states\x\valuations]_\req$, consider the configurations that are reached by at least one $r$-consistent run, and mark one of them $m_R$, if at least one exists, along with one $r$-consistent run $\rho_R$ leading to the configuration $m_R$.
 
  Let $r'$ be the aspiring resolver that, when reading a letter $a$, considers the region $R$ of the current configuration, and  follows what $r$ does when reading $a$ after the marked $r$-consistent run $\rho_R$.
  We set $r'(\rho,a)\eqdef r(\rho_R,a)$ where $R$ is the final region of the prefix-run $\rho$. 
  Note that $r'$ is well defined since it always follows transitions consistent with some $r$-consistent run
  and can therefore only visit marked regions.

We claim that $r'$ is indeed a resolver. Towards a contradiction, assume that it is not a resolver, that is, there is some word $w\in \lang(\T)$ for which $r'$ builds a rejecting run.
As $\T$ is a safety automaton, we can consider the last configuration $(s,\nu)$ along this run from which the remaining suffix $au$ of $w$ can be accepted
\footnote{
The fact that a rejecting run produced by a non-resolver must ultimately
reach a configuration that cannot accept the remaining word
also holds for TAs over finite words.
However, this is \emph{not} the case for infinite words defined via reachability acceptance.
}.

Suppose that $\rho$ is the prefix of the run
built by $r'$ on $w$, which ends in
$(s,\nu)$  and let $\tau = r'(\rho,a)$ be the $a$-transition chosen by $r'$.
We know that $\tau$ leads from $(s,\nu)$ to some configuration $(s',\nu')$ 
from where $u$ is not accepted.
By definition of $r'$, there must be a marked configuration $m_R\req (s,\nu)$ 
reached by some run $\rho_R$ %
from which $r$ chooses the same $a$-transition $\tau$. 
By \cref{lem:runtree-eq} there must be a word $au'$ so that
the run-tree on $au$ from $(s,\nu)$ is equivalent to that on $au'$ from $m_R$.
This means that $au'\in\lang(m_R)$ and, as $r$ is a resolver,
there must be an accepting run that begins with a step $(m_R)\step{\tau}(m'_R)$.
We derive that $u$ also has an accepting run from $(q,\nu)$ that begins with $\tau$, contradicting the assumption that $(q, \nu)$ is the last position on the run $r'$ built on $w$ so that its suffix can be accepted.
Therefore, $r'$ is indeed a resolver.
\end{proof}

Using a technique of \cite{BL22} we can show the existence of region-based resolvers also for HD timed reachability automata.

\begin{lemma}
    \label{lem:positional-reach-resolvers}
Every history-deterministic TA with reachability acceptance has a region-based resolver.
\end{lemma}

\begin{proof}
    Given a TA $\T$ with a reachability condition, we call a configuration $(s,\nu)$ \emph{almost final} if
    it is universal
    ($\lang(s,\nu)=\Sigma_T^\omega$)
    and
    Player~2 wins the letter game on $\T$ starting at $(s,\nu)$. 	
    In particular, $(s,\nu)$ is almost final if $s$ is a final state. 

We first argue that this is a property of regions, rather than just configurations:
if $(s,\nu)\req(s,\nu')$ and $(s,\nu)$ is almost final then also $(s,\nu')$ must be almost final.
Indeed, first observe that $(s,\nu')$ is universal by \cref{lem:runtree-eq}. 

To see why there must also be a resolver from $(s,\nu')$,
we observe that the letter game starting from any universal configuration
can be turned into an equivalent timed reachability game.
This is because by the universality assumption, Player~1 can only win by
keeping her opponent away from accepting configurations forever.
The construction uses one extra clock that is used in particular to prevent Player~2 from making delay steps.
For such games, both players have region-based winning strategies (see \cref{lem:tg-region-strategies}).
Consequently, Player~2 can use the same region-based strategy in the letter games
from $(s,\nu')$ and $(s,\nu)$, which exists by our initial assumption.

We now argue that in the letter game for $\T$,
any resolver must reach an almost final configuration as soon as Player~1 has played a word $u$ for which some run $\rho$ reaches an almost final configuration. Indeed, all continuations of $u$ are in $\lang(\T)$, so the configuration $(s,\nu)$ reached by the resolver must accept all continuations, and the resolver must represent a winning strategy in the letter game from $(s,\nu)$.

We now turn $\T$ into a new timed automaton $\T'$ over finite words
with a single new accepting state that can be reached (via some new letter $\$$)
exactly from all almost-final configurations.
This is well defined because almost-final is a property of regions and therefore can be expressed by a clock constraint.
Note that $\T'$ is history-deterministic because a resolver for $\T$ is also a resolver for $\T'$. This is because a resolver for $\T$ gives on a word $u$ a run to an almost final configuration, whenever there exists one.
Since $\T'$ is a history-deterministic TA on finite words, there exists a region-based strategy $\strat_1$ in the letter game for $\T'$  by \cref{lem:region-resolver}. 
We combine $\sigma_1$ with a region-based strategy $\strat_2$ in the letter game from almost final configurations in $\T$.
This yields a region-based strategy for Player~2 in the letter game for $\T$ which plays according to $\strat_1$
while not almost final, and $\sigma_2$ from then onwards.
This strategy must be a resolver:
If Player~2 provides a word $u$ such that, for no finite prefix there is a run ending in an almost final configuration,
then $u$ is not in $\lang(\T)$ and Player~2 wins.
If conversely, some prefix of $u$ admits a run ending in an almost-final configuration,
then $\sigma_1$ guarantees that Player~2 must be in such configuration after reading this prefix.
From here $\strat_2$ ensures that she wins whatever Player~1 plays.
\end{proof}

We can now use the region-based solver to determinize history-deterministic safety and reachability TA.

\begin{theorem}
\label{thm:safety-determinisable}
Every history-deterministic safety and reachability TA is equivalent to a deterministic TA.
\end{theorem}

\begin{proof}
Consider a history-deterministic TA $\T=(Q,\iota,C,\Delta,\Sigma, \acc)$, with a region-based resolver (as in \cref{lem:region-resolver}) $r$, and let $R$ be the region graph of $\T$. Define $\T'=(Q,\iota,C,\Delta',\Sigma, \acc)$ where $(q,g\wedge z,a,X,q')\in \Delta'$  for $z$ a guard defining a region of $R$, that is, a guard that is satisfied exactly by valuations in $R$, if $(q,g,a,X,q')\in \Delta$ is the transition chosen by $r$ in the region defined by the guard $z$.
 In other words, $\T'$ is $\T$ with duplicated transitions guarded so that a transition can only be taken from a region from which $r$ chooses that transition. 
 Observe that $\T'$ is deterministic: the guards describing regions are mutually exclusive,  therefore the guards of any two transitions from the same state over the same letter have mutually exclusive guards.

As runs of $\T'$ corresponds to a run of $\T$ with added guards, $\lang(\T')\subseteq \lang(\T)$. Conversely, if $w\in \lang(\T)$, then its accepting run consistent with $r$ is also an accepting run in $\T'$, since each transition along this run, being chosen by $r$, is taken at a configuration that satisfies the additional guards in $\T'$. We can therefore conclude that $\lang(\T)=\lang(\T')$.
\end{proof}

While this determinization procedure preserves the state-space of the automaton, %
it multiplies the number of transitions (or the size of guards) by the size of the region abstraction. Thus, while history-deterministic reachability and safety TA are no more expressive than deterministic ones, they could still be more succinct, 
when counting transitions and guards.
Notice that this construction requires guards of the form $x-y\triangleleft c$, where $y\in\clocks$ is a clock (aka diagonal constraints). 
These can be removed using the construction of \cite{BPDG98} by adding additional states to obtain an equivalent deterministic TA without diagonal constraints.

\subsection{CoB\"uchi HD TA are not determinizable}
\label{sec:expressivity:d-hd}
We show that history-deterministic timed automata are strictly more expressive than deterministic ones, for coB\"uchi acceptance and above.
For simplicity let's first assume that our definition of timed words/languages excludes Zeno words.
The case were these are allowed is a straightforward adjustment that we comment on at the end of the section.
The rest of this subsection is a proof of the following theorem.

\begin{theorem}
    \label{thm:d-hd}
    
    History-deterministic co-B\"uchi TA are more expressive than deterministic parity TA.
\end{theorem}
\begin{proof}

The following timed language $L$ over the singleton alphabet $\Sigma=\{a\}$
is recognized by a one-clock history-deterministic co-B\"uchi automaton
yet not by any deterministic Parity timed automaton.
In words, $L$ asks to eventually see events $a$ at unit distance. Formally,
\[
    L~\eqdef~\left\{(a,t_0)(a,t_1)... \mid \exists i\in\N.\quad \forall n\in\N.\quad 
    \exists j>i.\quad t_{j}-t_i =n\right\}.
\]

We first show that $L$ is recognized by the history-deterministic timed $\omega$-automaton in Fig. \ref{fig:introexa}. This automaton has an initial rejecting state $q$, from where there is a non-deterministic choice to either remain in this state or transition to an accepting state $q'$, which resets the unique clock. There are two transitions to stay in the accepting state: one enabled when the clock value is smaller than $1$, and one enabled at clock value $1$, which also resets the clock. If the clock value grows larger than $1$, the only enabled transition goes back to the initial state. Since this is a co-B\"uchi automaton, an accepting run must eventually remain in the accepting state.
\begin{figure}[t]
    \centering
    \includegraphics{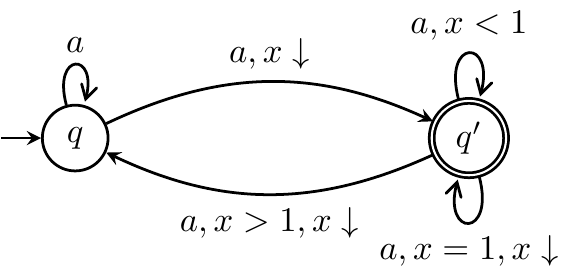}
    
    \caption{A history-deterministic timed co-B\"uchi automaton for $L$. The state $q'$ has priority $0$, i.e. is accepting, while the state $q$ has priority $1$.}
    \label{fig:introexa}
\end{figure}

First, this automaton recognizes $L$: if $w\in L$ then there is an accepting run that moves to state $q'$ at time $t$, where it then remains since the clock $x$ is reset at the occurrence of each event $(a,t+n)$ for $n\in \mathbb{N}$, so the clock value never grows larger than $1$. Conversely, a word accepted by this automaton has a run that eventually moves to $q'$ at a time $t$, and then remains in $q'$. For the run to stay in $q'$, it must reset $x$ at every time-unit after $t$, so $(a,t+n)$ must occur in the word for all $n\in \mathbb{N}$, that is, the word is in $L$.

We now argue that this automaton is also history-deterministic.
Given a finite word read so far and a new letter $a$ at time $t_{\textit{new}}$, the resolver identifies the earliest time $t_{\textit{early}}$ such that $a$ has so far occurred at time $t_{\textit{early}}+n$ for all integers $n$ such that $t_{\textit{early}}+n\leq t_{\textit{new}}$. Let $r$ be the function that maps a run $\rho$ ending in $q$ to $q'$ if $t_{\textit{new}}=t_{\textit{early}}+m$ for some integer $m$, and otherwise to the only other available transition.

We claim that this is indeed a resolver.
If $w\in L$ then there is an earliest time $t$ such that $(a,t+n)$ occurs in $w$ for all integers $n$. Since $t$ is minimal, eventually the resolver $r$ will make its choice whether to move to $q'$ over a letter $(a,t_{\textit{new}})$ based on whether $t_{\textit{new}}=t+m$ for some integer $m$. Since time progresses and $(a,t+n)$ occurs in $w$ for all integers $n$, the run will eventually transition to $q'$ at a time $t+m$ for some $m$. From there, since $(a,t+n)$ occurs in $w$ for all integers $n$, the run over $w$ remains in $q'$ and is therefore accepting.

\begin{figure}[t]
\centering
    \includegraphics{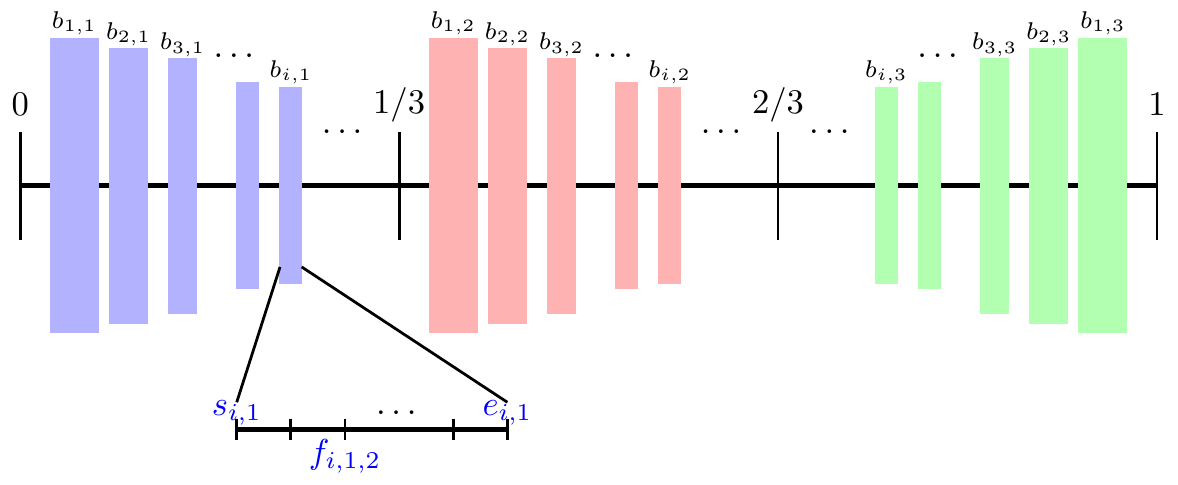}
\caption{Blocks 
    $\block{i}{1}$,
    $\block{i}{2}$, and $\block{i}{3}$ 
    within a unit time interval 
    are displayed in blue, red and green respectively.
    Each block $\block{i}{j}$ has equally spaced out events $\tick{i}{j}{0}\ldots \tick{i}{j}{r}$ where $r$ is at most the number of regions.
}
\label{fig1}
\end{figure}

\bigskip
It remains to be shown that $L$ is not recognized by a deterministic timed automaton.
Suppose towards a contradiction that $L$ is recognizable by some deterministic Timed Automaton $D$
with Parity acceptance. Let $r$ be the number of its regions.

We will construct two words $w$, $w'$ such that $w\in L$ and $w'\not\in L$,
so that the run of $D$ on $w$ is region equivalent to the run on $w'$.
The two words can only differ in the timing of events since there is only one letter in the alphabet.

Both words will be constructed on the fly, according to the following schema.

Consider the intervals and fractional values in Fig.~\ref{fig1};
there are infinitely many disjoint intervals,
$\block{i}{j} = [\bstart{i}{j}, \bend{i}{j}]$
so that all $\block{i}{1}$
have start and endpoint strictly between $0$ and $\blimit$
and are increasing, i.e., $\bstart{i+1}{1} > \bend{i}{1}$ for all $i$.
Similarly,
$\block{i}{2} \subseteq [\blimit,\frac{2}{3}]$,
and $\bstart{i+1}{2} > \bend{i}{2}$
for all $i$.
The third sequence of intervals 
$\block{i}{3} \subseteq [\frac{2}{3},1]$
have start and endpoint strictly between $\frac{2}{3}$ and $1$
and are \emph{decreasing}:
$\bend{i+1}{3} < \bstart{i}{3}$ for all $i$.
Each interval $\block{i}{j}$ contains equi-distant values 
$
\tick{i}{j}{0},
\tick{i}{j}{1},
\ldots,
\tick{i}{j}{r}
$
starting at $\tick{i}{j}{0}=\bstart{i}{j}$, where $r$ is the number of regions.

We step-wise construct $w$ (and $w'$) together with the run of $D$ on it.
In every integral interval from $i-1$ to $i$ we place events as follows.

\begin{itemize}
    \item start with a delay of $\tick{i}{1}{1}$,
    followed by a discrete event $a$, then delay of  $\tick{i}{1}{2}-\tick{i}{1}{1}$ followed by $a$,
    and so on. This induces a run of $D$ on the prefix constructed and we continue constructing the prefix until the induced run closes a cycle in the region graph. Formally, there exists time $\tick{i}{1}{k}$ 
    and $\tick{i}{1}{k+\ell}$ such that the run is in configuration 
    $(\confstate{i}{1}{k}, \confval{i}{1}{k})$ before reading the $k$th $a$ and in configuration $(\confstate{i}{1}{k+\ell}, \confval{i}{1}{k+\ell})$ at time $\tick{i}{1}{k+\ell}$ such that 
    $(\confstate{i}{1}{k}, \confval{i}{1}{k}) \req (\confstate{i}{1}{k+\ell}, \confval{i}{1}{k+\ell})$ at time $\tick{i}{1}{k+\ell}$.
    We denote by $L_i$ the run between $\tick{i}{1}{k}$ and $\tick{i}{1}{k+\ell}$. 
    
    \item Now we force  the automaton to close the same cycle $L_i$, but with all events occurring at times in the interval $\block{i}{2}$ (respectively $\block{1}{2}$) in $w$ (respectively $w'$).
    This can be done by adding a time delay by $\bstart{i}{2}-\tick{i}{1}{k+\ell}$ 
    in $w$ followed by an event $a$ at times $\tick{i}{2}{\ell'}$ for all $\ell'\leq \ell$. 
    We prove this formally in Lemma \ref{lem:safe}.

    \item Finally we force the automaton to close the same cycle $L_i$ once more, with all times
        in interval $\block{i}{3}$.
        This can be done by adding a time delay $\bstart{i}{3}-\tick{i}{2}{\ell}$ followed by events at times
        $\tick{i}{3}{1},
        \tick{i}{3}{2},\ldots
        \tick{i}{3}{\ell}
        $.
      We prove the correctness of the construction in Lemma \ref{lem:safe}.
  
\end{itemize}

Consider the cycle $L_i$ in the region graph obtained in step $1$ above in the interval $[i-1,i]$, between $\tick{i}{1}{k}$ and $\tick{i}{1}{k+l}$. Note that the $k$ and $\ell$ depends on $i$. However, we write $k$ and $\ell$ 
without
the dependency on $i$ 
as we only reason about loops within an integral interval.
The duration of the loop, denoted by $\dur{L_i}$ is $\tick{i}{1}{k+\ell}-\tick{i}{1}{k}$.
An important observation is that $\dur{L_i}\leq \bend{i}{j}-\bstart{i}{j}$ as the loop occurs within the interval $\block{i}{1}$. 
  
\begin{lemma}
	\label{lem:safe}
Let $\nu_i$ and $\nu'_{i}$ be the configurations reached by the run of $D$ at times $i-1+\tick{i}{1}{k}$ and $i-1+\tick{i}{1}{k+\ell}$.
Then $1-\maxfractof{\nu_i+d_{ij}}\geq \dur{L_i} $, where $d_{ij}=\bstart{i}{j}- \tick{i}{1}{k}$ for $j\in \{2,3\}$.

Furthermore, let $\nu_{ij}$ be the configuration reached by the run of $D$ at time $i-1 + \tick{i}{j}{1}$, where $j=\{2,3\}$. The cycle $L_i$ is executable from $\nu_{ij}$.

\end{lemma}

\begin{proof}
	We prove this lemma by induction on $i$.
	The case $i=1$ is easy to see since
	$\maxfractof{\nu_1+d_{1j}} \leq \bstart{1}{j}$ and therefore $1-\maxfractof{\nu_1+d_{1j}}\geq 1-\bstart{1}{j}\geq \bend{1}{j}-\bstart{1}{j}\geq \dur{L_1}$.
	
    Furthermore, $\nu_{12}=\nu'_{1}+d$, where $d=\bstart{1}{2}-\tick{1}{1}{k+\ell}\leq 1-\tick{1}{1}{k+\ell}\leq 1- \maxfractof{\nu'_1}$. Therefore, by Proposition~\ref{lem:paths}.\ref{lem:paths-short-delay},  $\nu_{12}\req \nu'_1\req \nu_1$. 
	For $\nu= \nu_{1}$ and $\nu=\nu_{12}$, $1-\maxfractof{\nu}> \bend{1}{3}-\bstart{1}{3}>\dur{L_1}$
	as $1>\bend{1}{3}$ and $\maxfractof{\nu}<\bstart{1}{3}$.
	By applying Proposition~\ref{lem:paths}.\ref{lem:paths-short-path}, $L_1$ is executable from $\nu_{12}$ and ends in a configuration $\nu'_{12}\req \nu_{12}$. 
	
    The configuration $\nu_{13}$ equals $\nu'_{12}+d'$, where $d'= \bstart{1}{3}- \tick{1}{2}{\ell}< 1-\maxfractof{\nu'_{12}}$ as $\maxfractof{\nu'_{12}}\leq  \tick{1}{2}{\ell}$.
    Proposition~\ref{lem:paths}.\ref{lem:paths-short-delay} gives $\nu_{13}\req \nu_{12}$, and $1-\maxfractof{\nu_{13}}\geq \bend{1}{3}-\bstart{1}{3}\geq \dur{L_1}$. By Proposition~\ref{lem:paths}.\ref{lem:paths-short-path}, we can conclude that $L_1$ is executable from $\nu_{13}$.
	
	To prove the inductive case, we bound the value of $\maxfractof{\nu_i + d_{ij}}$ for $j\in\{2,3\}$. Consider a clock $x\in \clocks$ and the last time when it was reset. Either it was never reset or the reset occurred at time $\tick{i'}{j'}{k'}$. 
	For a clock that is never reset, the fractional part of its value at $\nu_i$ will be $\tick{i}{1}{k}$. 
	If the clock was last reset within some blue block, i.e, at time $i'-1+\tick{i'}{1}{k'}$,  then either $i'<i$ (corresponds to previous blue blocks), or $k'<k$ (corresponds to previous ticks within the current blue block).
	In both cases, the $\fractof{x}= \fractof{\tick{i}{1}{k}-(i'-1 +\tick{i'}{1}{k'})}\leq  \tick{i}{1}{k}$. 
	
	Note that any reset to clock $x$ in a previous red block must also be reset again in the corresponding green block as the runs in the red and green block are the same by construction. 
	For a clock $x$ last reset in some previous green block, i.e, at time $i'-1+ \tick{i'}{3}{k'}$, $\fractof{x}=\fractof{(i-1+ \tick{i}{1}{k}) - (i'-1 +\tick{i'}{3}{k'})}
	= 
	\tick{i}{1}{k}+(1- \tick{i'}{3}{k'})$.
	Furthermore, $\tick{i'}{3}{k'}>\bstart{i-1}{3}$ as $i'\leq i$. Therefore, $\fractof{x}\leq 1+\tick{i}{1}{k}-\bstart{i-1}{3}+1$ which bounds $1-\fractof{x}\geq \bstart{i-1}{3}-\tick{i}{1}{k}$.
	Combining all the possibilities for clock resets, we obtain $1-\maxfractof{\nu_i}\geq \bstart{i-1}{3}-\tick{i}{1}{k}$.

	It is easy to see that for $j\in \{2,3\}$, $\maxfractof{\nu_i+d_{ij}}\leq \maxfractof{\nu_i}+d_{ij}$. Therefore, $1-\maxfractof{\nu_i+d_{ij}}\geq 1-\maxfractof{\nu_i} - d_{ij}\geq 1- (\tick{i}{1}{k}-\bstart{i-1}{3}+1) - (\bstart{i}{j}- \tick{i}{1}{k})\geq \bstart{i-1}{3} - \bstart{i}{j}\geq \bend{i}{j}-\bstart{i}{2}$. The last step follows from the fact that $\bend{i}{j}\leq \bstart{i-1}{3}$ for $j\in\{2,3\}$. Note that the duration of the loop $L_i$ is less than $\bend{i}{j}-\bstart{i}{j}$ and thus completes the proof for first part of the lemma.

    We now show that $L_i$ is executable from $\nu_{i2}$ and $\nu_{i3}$. 	First, $\nu_{i2}\req \nu_i$ and $\nu_{i3}\req \nu_{i2}$ by repeated application of Proposition~\ref{lem:paths}.\ref{lem:paths-short-delay}. This is similar to the argument in the base case.
	We just showed that $1-\maxfractof{\nu_{i2}}> \bstart{i-1}{3}-\bstart{i}{2}> \bend{i}{2}-\bstart{i}{2}= \dur{L_i}$. The same argument holds for $\nu_{i3}$ as well. Also, $\maxfractof{\nu_{i}} \leq 1- \bstart{i-1}{3} + \tick{i}{1}{k}$ and hence $1-\max\nu_{i}\geq  \bstart{i-1}{3}-\tick{i}{1}{k}<\bend{i}{3}-\bstart{i}{3}<\dur{L_i}$.
    Therefore, by Proposition~\ref{lem:paths}.\ref{lem:paths-short-path}, $L_i$ is executable from both $\nu_{i2}$ and $\nu_{i3}$. 
\end{proof}

Notice that the so-constructed word $w$ is not in $L$ because all $\block{i}{j}$ are disjoint.
The word $w'$ will be constructed almost the same way, with the only exception that
the first repetition of the cycle is moved not to $\block{i}{2}$
but always the same interval, $\block{1}{2}$. 
It is easy to see that Lemma \ref{lem:safe} can be modified where $\block{i}{2}$ is replaced everywhere by $\block{1}{2}$. 
In particular this means that
$w$ contains an event at time $n+\bstart{1}{2}$ for any $n\in\N$, and thus must be contained in $L$.
Therefore, $D$ has an accepting run on $w'$ but the run on $w'$ visits the same sequence of states as the run of $D$ on $w$. Therefore, $D$ must accept $w$ as well, which is a contradiction proving that $L$ is not accepted by any deterministic Timed Automaton with Parity acceptance.

This concludes the proof of \cref{thm:d-hd}.
\end{proof}

\begin{remark}
    \label{rem:d-hd-zeno}

In case our definition of timed words does allow for Zeno words, the argument above can be adjusted as follows.
Instead of $L$, we consider the language $L\cup Z$,
where $Z$ denotes the set of all Zeno words.
In fact the automaton depicted in \cref{fig:introexa} recognizes $L\cup Z$. However, the candidate resolver $r$ proposed above may wait for ever in state $q$ for the ``oldest'' fractional timestamp to re-appear while reading in a Zeno word. This results in a rejecting run on a word in the language, and so the $r$ is no resolver.
To fix this issue we adjust the automaton to the one depicted in \cref{fig:introexaa}.
This is still a co-B\"uchi timed automaton, where the acceptance condition is via transitions: we want to finitely often use the edges that change state. Clearly this can be translated into a TA with state-based acceptance as before. The recognized language is still $L\cup Z$, the same as before.
Note that runs that stay in $q$ for ever are accepting Zeno words.
For this automaton, the candidate resolver $r$ (and its justification) is the same as above, except the run built for Zeno-continuations when in state $q$ is actually accepting. The new automaton is therefore history-deterministic. A proof that $L\cup Z$ is not recognizable by any Parity DTA remains the same because the words constructed in the proof are non-Zeno.

\begin{figure}[t]
    \centering
    \includegraphics{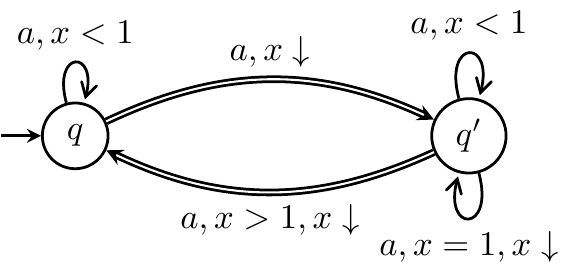}
    
    \caption{A history-deterministic timed co-B\"uchi automaton for $L\cup Z$ (the double edges are rejecting, all other accepting).}
    \label{fig:introexaa}
\end{figure}
\end{remark}

\subsection{HD TA are less expressive than fully non-deterministic TA}
\label{sec:expressivity:hd-nd}

\begin{theorem}
    \label{thm:hd-nd}
    History-deterministic reachability TA are less expressive than non-determi\-nis\-tic parity TA.
    
\end{theorem}

\begin{figure}[t]
	\centering
    \includegraphics{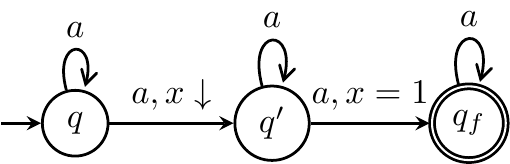}
	\caption{A non-deterministic timed reachability automaton for $L'$.}
	\label{fig:nd}
\end{figure}
 
\begin{proof}
	
The following language $L'$ over the singleton alphabet $\Sigma=\{a\}$
is recognized by a one-clock non-deterministic TA with reachability acceptance but not by any history-deterministic Parity TA.
In words, $L'$ asks to see two events $a$ at unit distance. Formally,
\[
	L' \eqdef \left\{(\sigma_0,t_0)(\sigma_1,t_1)... \mid \exists i,j\in\N.\quad t_{j}-t_i =1 \text{ and } \sigma_i=a \text{ and } \sigma_j = a\right\}.
\]

The non-deterministic TA shown in Figure~\ref{fig:nd} accepts the language $L'$ by guessing positions $i$ by reading an $a$, resetting a clock $x$ and checking that it sees an $a$ at distance $1$.

Assume towards a contradiction that there exists a HD TA $H$ with $k$ clocks and maximum constant in guards $\cmax{x}$, that recognizes $L'$. 
For all $i\le k$ consider the finite word
\[
    w_i= \left(a,\frac{1}{k+1}\right)\cdots \left(a,\frac{k+1}{k+1}\right) \left(a, 1+\frac{i}{k+1}\right) %
\]
that sees $k+1$ equi-distant events in the interval $[0,1]$ and then repeats the $i$th fractional value
in the next integral interval.
Note that $w_i\in L'$, for all $i\leq k$
and so the resolver gives a run on all such words. 
Note that the prefix up to time $1$ is the same on all $w_i$ and therefore the resolver gives the same run on the prefix until time $1$.
Consider the configuration $\nu$ reached by the resolver after reading the prefix up until and including
the discrete event $(a,1)$.
Since $H$ has $k$ clocks and $k+1$ events $a$, there exists $j\le k+1$ 
such that either no clock is reset at the $j$th $a$-event, or $j\le k$ and the clocks reset at the $j$th $a$ are reset again at a later $a$-event. 
This implies that the configuration $\nu$ reached at the end of the block is such that $\nu(x) \neq 1-\frac{j}{k+1}$ holds for all clocks $x$.
To see this note that if $x$ was reset before the $j$th $a$, $\nu(x)\ge 1-\frac{j-1}{k+1}$ and if $x$ was reset after the $j$th $a$, $\nu(x) \le 1-\frac{j+1}{k+1}$. 
It follows that
\[\nu+\frac{j}{k+1} \req \nu+\frac{j}{k+1} +\left(\frac{1}{2(k+1)}\right)\]
because
for all clocks reset before the $j$th $a$, the value of $x$ in both valuations is $>1$ and for all clocks reset after the $j$th $a$, both valuations will satisfy $x<1$. 

Finally, let's take the word
\[
    w'= \left(a,\frac{1}{k+1}\right)\left(a,\frac{2}{k+1}\right)\cdots \left(a,\frac{k+1}{k+1}\right) \left(a, 1+\frac{j}{k+1}+\frac{1}{2(k+1)}\right)
\]
Clearly $w'$ is not in $L'$. 
However, $H$ must have a run on $w'$ which follows the accepting run of $H$ on $w_j$.
The final step in this run can be executed because the two runs end up in equivalent configurations.
A contradiction.
\end{proof}

By \cref{thm:hd-nd,thm:d-hd} we thus conclude that the classes of languages accepted by deterministic, history-deterministic and non-deterministic TAs are all different.

	\section{Timed Games and Composition}
	\label{sec_games}
	\label{sec:games}
	In this section we consider several games played on (LTSs of) timed automata and how they can be used to decide classical verification problems.
We focus on turn-based games, although our techniques can be generalised to concurrent ones. We first look at language inclusion, then synthesis, and finally we consider good-for-games timed automata, that is, automata that preserve the winner when composed with a game and show that good-for-gameness and history-determinism coincide for both reachability and safety timed automata.

\begin{definition}[Timed Games] A \emph{timed game} is a turn-based two player
    game played on the configuration graph of a timed automaton.
    Formally, it consists of an arena $\gsgame=(Q,\iota,  C, \Delta,\Sigma,L)$,
    which is a TA except that the set $Q$ of states is partitioned into those belonging to Player~1 ($Q_1$)
    and those belonging to Player~2 ($Q_2$), and that the acceptance condition $L$ is a timed language
    (not some set of acceptable runs).
    
    Configurations are defined as for TA. %
    A timed game starts in the initial configuration $(\iota,\vec{0})$ and proceeds so that each round $i$ from configuration $c_i$, 
the owner of the current state first advances time with a delay $d_i\in \nnreals$
and follows it up with a discrete step from $c_i+d$ leading to a successor configuration $c_{i+1}$ and emitting a letter $a_i\in\Sigma$.
An infinite play is winning for Player~2 iff the word $d_0a_0d_1a_1\dots$ produced is in $L$. %

A \emph{timed parity game} is a timed game where the acceptance condition $L$ is given by a timed parity automaton. We assume w.l.o.g., that the configuration graph of the game and that of the parity automaton defining $L$ is the same, so that a timed parity game is simply given as a timed game as above, and a colouring $\priority:Q\mapsto\N$.
\end{definition}

\begin{remark}
    Notice that the definition of timed games above enforces that delays and discrete letters are emitted alternatingly. In particular, every play uniquely introduces an (infinite) timed word.
    The definition does not explicitly forbid these words to be Zeno (as for instance \cite{CHP08} do)
    nor requires that only Player~1 advances time.
    To enforce that only Player~1 enforces delay, we can add a clock that is reset at every transition and is checked for $0$ at every outgoing transition from a Player~2 vertex.    
    To handle Zeno words, various winning conditions have been proposed that avoid that one player can always win by delaying. For example, as Zeno words are always winning for Player~2, or the player proposing $0$ delay infinitely often loses. We refer to the introduction of \cite{DFHM03} for a summary of the different frameworks to handle Zeno words in the winning conditions.
    
\end{remark}

We will occasionally use the fact that one can effectively determine the winner of a timed parity game.

\begin{theorem}[Theorem 5, \cite{DFHM03}]
    \label{lem:TPG}
    \label{lem:tg-region-strategies}
        Consider a timed parity game with $m$ states, $n$ clocks, maximal constant $k$ and $c$ priorities.
    \begin{enumerate}
        \item If Player~$i$ has a winning strategy then also one that is region-based, i.e., makes the same choices from all configurations in the same region.
        \item The game can be solved in time $\mathcal{O}((m\cdot n!\cdot 2n\cdot (2k+1)n\cdot m\cdot c)^{(c+1)})$, therefore, in \exptime.
    \end{enumerate}
\end{theorem}

\subsection{Language Inclusion and Fair Simulation Games}
The connection between history-determinism and fair simulation, established in \cref{thm:triangle},
allows to transfer decidability results to history-deterministic TA.
Let's first recall that simulation checking is decidable for timed automata using a region construction
\cite{TAKB96}.
This paper precedes the notion of fair simulation (restricting Player~1 to fair runs)
and is thus only applicable for safety conditions. However, the result holds for more general parity acceptance
(for which each state is assigned an integer priority and where a run is accepted if the highest priority it sees infinitely often is even).

\begin{theorem}
    \label{lem:GFGTA:fairsim-ub}
    \label{lem:GFGTA:fairsim}
    Checking fair simulation is in \EXPTIME\ for parity timed automata.
\end{theorem}
\begin{proof}
For membership in \EXPTIME, it suffices to observe that the simulation game can be presented as timed parity game with suitable bounds, then apply \cref{lem:TPG}.

Assume w.l.o.g.\ that both configurations to check are in the same timed automaton
$\auta = (\states,\iota,\clocks,\transitions,\Sigma,\priority)$ with $c=\card{\priority(Q)}$ many colours.

The game is played on the product, with each player moving on their coordinate according to the TA semantics.
The arena is bipartite and enforces alternation; Player~1 states are $Q_1=(Q\x Q)$, Player~2 states are $(\transitions\x Q)$.
The latter record Player~1's choice of discrete transition $t\in\transitions$.
Player~1 first announces a delay (which affects all clocks and she remains in control) or the transition she wants.
The latter choice resets some new clock and moves to a Player~2 owned state. 
Player~2 then has to move according to a transition carrying the same label as $t$,
and updating both components along the two chosen transitions.
This way, both players jointly construct on a timed word read in both copies.

The winning condition (parameter $L$ for the timed parity game) encodes
that if Player~1's run is accepting then so is that of Player~2.
This can be achieved with a blow-up of the state space that is exponential only in $c$, the number of colours.
Essentially, to implement the implication, 
\begin{enumerate}
    \item increment all colours in Player~1's copy by one to negate the set of accepted runs;
    \item interpret the parity colouring(s) as Rabin chain condition with $c/2$ many indices\footnote{More precisely, if $c$ is even then both Rabin chain conditions have size $c/2$ and otherwise, if it is odd, they both have size $\lceil c/2\rceil$.}.
    \item create a new Rabin acceptance condition for the product that implements the implication:
        for every index set $(F_i,I_i)$ of Player~1, of finite and infinitely occurring states, resp.,
        introduce a new Rabin pair $(F_i\x\states, I_i\x\states)$ and similarly,
        for every index set $(F'_i,I'_i)$ of Player~2, introduce a new Rabin pair $(\states\x F'_i, \states\x I'_i)$. The size of the resulting index set is $c$. 
    \item Finally, turn the whole game with Rabin acceptance back into one with parity acceptance
        using least appearance records \cite{Gurevi82,Dziemb97}.
        This results in a timed parity game with $m=\card{\states^2+\transitions\x\states} \cdot c!$
        states and $2c+1$ colours.
\end{enumerate}
The claim of \exptime~ membership now follows by applying \cref{lem:TPG} as the obtained complexity can be bounded by $\mathcal{O}((m\cdot n!\cdot 2n \cdot (2k+1)n \cdot 2c+1)^{(2c+2)})$. The base in the above expression is dominated by $n!\cdot c!$ which can be bounded by $2^{n\cdot c}$ by using Stirling's approximation, which gives an exponential bound on the time complexity. 
\end{proof}
\begin{corollary}\label{cor:inclusion}
    Timed language inclusion is decidable and \EXPTIME-complete
    for history-deterministic TA.
    More precisely, given 
    a TA $\lts$ with initial state $q$
    and a history-deterministic TA $\lts'$ with initial state $q'$,
    checking if $q\langincl q'$ holds is \EXPTIME-complete.
\end{corollary}
\begin{proof}
    As $\autb$ is history-deterministic and by
    \cref{thm:triangle},
    we have 
    $q\langincl q'$ if, and only if,
    $q\fairsim q'$.
    The result follows from 
    \cref{lem:GFGTA:fairsim}.
\end{proof}

A matching lower bound holds even for safety or reachability acceptance, assuming that constants in transition guards are given in binary.
\begin{lemma}
    \label{lem:GFGTA:fairsim-lb}
    Checking fair simulation between TA is \exptime-hard already for reachability or safety acceptance, or over finite words.
\end{lemma}
\begin{proof}
This can be shown by reduction from \emph{countdown games} \cite{JLS2008},
which are two-player games $(Q,T,k)$ given by a finite set $Q$ of control states, a finite set $T \subseteq (Q \x \N_{>0} \x Q)$ of transitions, labelled by positive integers, and a target number $k\in\N$.
All numbers are given in binary encoding.
The game is played in rounds, each of which starts in a pair $(p,n)$ where $p\in Q$ and $n \le k$,
as follows.
First Player~1 picks a number $l \le k-n$, so that at least one $(p,l,p')\in T$ exists;
then Player~2 picks one such transition and the next round starts in $(p',n+l)$.
Player~1 wins iff she can reach a configuration $(q,k)$ for some state $q$.

Determining the winner in a countdown game is \exptime-complete \cite{JLS2008}
and can easily be encoded as a simulation game between two TAs $\auta$ and $\autb$ as follows.
Let $\auta$ be the TA with no clocks and unrestricted (guards are $\mathit{True}$) self-loops 
for the two letters $a$ and $e$;
the idea is that Player~1 proposes $l$ by waiting that long and then makes a discrete $a$-labelled move.
Then Player~2, currently in some state $p$ can update his configuration to mimic that of the countdown game,
and punish (by going to a winning sink) if Player~1 cheated or the game should end.
To implement this, $\autb$ has two clocks: one to store $n$ -- the total time that passed -- and one to store the current $l$, which is reset in each round.

Suppose Player~1 waits for $l$ units of time and then proposes $a$. 
Player~2, currently in some state $p$ 
will have 
\begin{itemize}
    \item 
$a$ and $e$-labelled transitions to a winning state with a guard that verifies that
there is no transition $(p,l,p')$. 
This can be done by checking the guard $\wedge_{i} (x_2\neq l_i)$ for 
all $l_i$ such that there exists $p'$ with an edge $(p,l_i,p')$ in the countdown game.
\item $a$-labelled transitions to a state $p'$, with a guard that verifies that 
some $(p,l,p')\in T$ exists, and which resets clock $x_2$.
\item $a$, and $e$-labelled transitions to a winning state guarded by $x_1>k$.
    This enables Player~2 to win if the global time has exceeded the target $k$.
\end{itemize}
The only way that Player~1 can win is by following a winning strategy in the countdown game
and by playing the letter $e$ once $\autb$ is in a configuration $(q,k)$.
Player~2 will not be able to respond.
\end{proof}

\subsection{Composition with Games}

Implicitly, at the heart of these reductions is the notion of composition: the composition of the game to solve with a history-deterministic automaton for the winning condition yields an equivalent game with a simpler winning condition. We say that an automaton is \emph{good-for-games} if this composition operation preserves the winner of the game for all games. While history-determinism always implies good-for-gameness, the converse is not necessarily true. While the classes of history-deterministic and good-for-games automata coincide for $\omega$-regular automata~\cite{BL19}, this is not the case for quantitative automata~\cite{BL21}, which can be good-for-games without being history-deterministic. We argue that for reachability and safety timed automata, good-for-gameness and history-determinism coincide.

Intuitively, the composition of a timed automaton $\T$ and a timed game $\gsgame$ 
consists of a game in which the two players play on $\gsgame$ while Player~2 must also build, letter by letter, a run of $\T$ on the outcome of the game in $\gsgame$.

\begin{definition}[Composition]
Given a  TA $\T$ and a timed game $\gsgame$ with winning condition $\lang(\T)$, the composition $\comp{\gsgame}{\T}$ consists of a game played on the product of the configuration spaces of $\T$ and $\gsgame$, starting from the initial state of both, in which, at each turn $i$, from a configuration $(c_i,c'_i)$, 
the owner of the current $\gsgame$-state chooses a time delay $d_i\in \nnreals$ advancing all clocks on both sides, and chooses a move in $\gsgame$ to a successor-configuration $c_{i+1}$, producing a letter $a_i$,
and then Player~2 chooses a transition over $a_i$ enabled at the current $\T$-configuration $c'_i$, leading to a successor-configuration $c'_{i+1}$. The game then proceeds from $(c_{i+1},c'_{i+1})$.

Player~2 wins infinite plays if the run built in $\T$ is accepting, and loses if it is rejecting or if she cannot move in the $\gsgame$-component.

\end{definition}
Observe that if Player~1 wins in $\gsgame$, then he also wins in $\comp{\gsgame}{\T}$ with a strategy that produces a word not in $\lang(\T)$ in $\gsgame$, as then Player~2 can not produce an accepting run in $\T$.

\medskip
Boker and Lehtinen \cite[Lemma 7]{BL21} show that for (quantitative) automata for which the letter-game is determined, (threshold) history-determinism coincides with good-for-gameness. The lemma is stated for quantitative automata, where thresholds are relevant; in the Boolean setting, it simply states that the determinacy of the letter game implies the equivalence of history-determinism and good-for-gameness. In our timed setting, a similar argument, combined with the determinacy of the letter game for safety and reachability TA, gives us the following.

\begin{theorem}\label{lem:gfg2hd}
	Let $\T$ be a safety or reachability TA. The following are equivalent: %
	\begin{enumerate}
		\item\label{item:hd} $\T$ is history-deterministic.
		\item\label{item:gfg} For all timed games $\gsgame$ with winning condition $\lang(\T)$, whenever  Player~2 wins $\gsgame$, she also wins $\comp{\gsgame}{\T}$.
	\end{enumerate} 

\end{theorem}

\begin{proof}
	(\ref{item:hd})$\implies$(\ref{item:gfg}) If $\T$ is history-deterministic, the resolver can be used as a strategy in the $\T$ component of $\comp{\gsgame}{\T}$. When combined with a winning strategy in $\gsgame$ that guarantees that the $\gsgame$-component produces a word in $\lang(\T)$, the resolver guarantees that the $\T$-component produces an accepting run, thus giving the victory to Player~2.
	
	(\ref{item:gfg})$\implies$(\ref{item:hd}) Towards a contradiction, assume $\T$ is not history-deterministic, that is, by determinacy of the letter game from~\cref{determinacy}, that Player~1 has a winning strategy $\strat$ in the letter game.
	Now consider the game $\gsgame_\strat$, without clocks or guards, in which positions, all belonging to Player~1, consist of the prefixes of timed words played by $\strat$, with moves $w\xrightarrow{(t,a)} w(t,a)$. %
	As $\strat$ is winning for Player~1, all maximal paths in $\gsgame_\strat$ are labelled by a timed word in $\lang(\T)$, so $\gsgame_\strat$ is winning for Player~2.
	
	We now argue that Player~1 wins $\comp{\gsgame_\strat}{\T}$ by interpreting Player~2's moves in the $\T$ component as her moves in the letter game, and choosing moves in $\gsgame$ mimicking the letter dictated by $\strat$. Then, if Player~2 could win against this strategy in $\comp{\gsgame_\strat}{\T}$, she could also win against $\strat$ in the letter game by interpreting Player~1's choices of letters as moves in $\gsgame$, and responding with the same transition as she plays in the $\T$ component of $\comp{\gsgame_\strat}{\T}$. Such a strategy is a valid strategy in the letter game on $\T$, and while it might not be winning in general, it is winning against $\strat$, contradicting that $\strat$ is a winning strategy for Player~1.
\end{proof}

This proof fails for acceptance conditions beyond safety and reachability, as it isn't clear whether timed B\"uchi and coB\"uchi automata define Borel sets. If this was the case then history-deterministic timed automata would be exactly those that preserve winners in composition with games, as is the case in the $\omega$-regular setting.

	\section{Deciding History-determinism}
	\label{sec_decision}
	\label{sec:decision}

Recall the letter game characterisation of history-determinism: Player~1 plays timed letters and Player~2 responds with transitions.  Player~2 wins if either the word is not in the language of the automaton, or her run is accepting.
As TA are not closed under complement, it isn't clear how to solve this game directly. Bagnol and Kuperberg~\cite{BK18} introduced  \textit{token games}, which are easier to solve, but which coincide with the letter game for various types of automata\cite{BK18,BKLS20,BL22}.

In this section we show that the simplest token game, called the one-token game, characterises history-determinism for safety and reachability LTSs, and therefore TAs. The case of reachability automata was already shown in~\cite{BL22}, but for self-containment we include the proof here, generalised to fair LTSs. This strengthens and simplifies the result of the conference version~\cite{HLT2022}, which used the two-token game to characterise the history-determinism of reachability automata.

\begin{definition}[1-token game~\cite{BK18}]
Given a fair LTS 
$\lts=(V,\Sigma,E)$ with initial state $s_0\in V$, the game $G_1(\lts)$ proceeds in rounds. At each round $i$:
\begin{itemize}
\item Player 1 plays a letter $a_i\in\Sigma$
\item Player 2 plays a transition $\tau_i$ in $\lts$
\item Player 1 plays a transition $\tau'_{i}$ in $\lts$
\end{itemize}
This way, Player~1 chooses an infinite word 
$w=a_0a_1\ldots$ and a run $\rho'=\tau'_{0}\tau'_{1}\tau'_{2}\dots$,
and Player~2 chooses a run $\rho=\tau_0\tau_1\dots$.
The play is winning for Player 1 if $\rho'$ is an accepting run over $t_0a_0...$ from $s_0$ but $\rho$ is not. Else it is winning for Player~2.

Given a TA $\T$, we write $G_1(\T)$ to mean the $1$-token game on the fair LTS induced by $\T$.
\end{definition}

\begin{remark}\label{determinacy} $G_1(\lts)$ and the letter game are determined for any fair LTS $\lts$ with any Borel-definable acceptance condition \cite{M1975}.
In particular, the letter game is determined for both safety and reachability TA $\T$. Indeed, the winning condition for Player~2 is a disjunction of the complement of $\lang(\T)$ and of the acceptance condition of $\T$. Then, as long as $\lang(\T)$ is Borel, by the closure of Borel sets under complementation and disjunction, the letter-game is Borel, and therefore determined, following Martin's Theorem~\cite{M1975}. If time is not required to diverge, then reachability timed languages and safety timed languages are clearly Borel. Since words in which time diverges are also Borel (they can be seen as the countable intersection of words where time reaches each unit time), this remains the case when we require divergence.
\end{remark}

$G_1(\lts)$ was shown to characterise history-determinism for a number of quantitative automata in~\cite{BL22}. In~\cref{lem:G1} below we show, using similar proof techniques, that this is also the case for all safety LTSs. The key observation is that for Player~2 to win the letter game, it suffices that she avoids mistakes. We then show that a winning strategy for her in $G_1(\lts)$ can be used to build such a strategy.

\begin{restatable}{lemma}{lemGone}\label{lem:G1}
Given a fair LTS $\lts$ with a safety acceptance condition, or interpreted over finite words, Player~2 wins $G_1(\lts)$ if and only if $\lts$ is history-deterministic.
\end{restatable}

\begin{proof}
If $\lts$ is history-deterministic then Player~2 wins $G_1(\lts)$ by using the resolver to choose her transitions. This guarantees that for all words in $L(\lts)$ played by Player~1, her run is accepting, which makes her victorious regardless of Player 1's run.

For the converse, if Player~2 wins $G_1(\lts)$ with strategy $\sigma$, consider the following family of \textit{copycat strategies} for Player~1: at first, Player~1 plays $\sigma$ and chooses the same transitions as Player~2.  If Player~2 ever chooses a transition $\tau$ from a configuration $c$ that is not \emph{language-maximal}, that is, moves to a configuration $c'$ that does not accept some word $w$ that is accepted by some other configuration $c''$ reachable by some other transition $\tau'$ from $c$, we call such a move \emph{non-cautious}, and Player~1 stops copying Player~2 and instead chooses $\tau'$. From there, Player~1 wins by playing $w$ and an accepting run on $w$ from $c''$. Since Player~2 wins $G_1(\lts)$, her winning strategy $\strat$ does not play any non-cautious moves against copycat strategies. 

Then, she can use $\strat$ in the letter-game, by playing as $\strat$ would play in $G_1(\lts)$ if Player~1 copies her transitions. This guarantees that she never makes a non-cautious move, and, in particular, for $\lts$ with a safety acceptance condition, never moves out of the safe region of the automaton unless the prefix played by Player~1 has no continuations in $L(\lts)$. This is a winning strategy in the letter-game, so $\lts$ is history-deterministic.
Similarly, for $\lts$ interpreted over finite words, the letter game is a safety game, and a strategy that never makes a non-cautious move remains in the safe region of this game, and is therefore winning.
\end{proof}

This argument does not work for reachability TA: it is no longer enough for Player~2 to avoid bad moves to win; she needs to also guarantee that she will actually reach a final state.

Given a fair LTS $\lts$ with a reachability acceptance condition, we call a state $q$ \textit{almost final} if, from $q$, there is an accepting run on all words and Player $2$ wins the letter game starting from $\lts^q$. In particular, every final state is also almost final.

\begin{lemma}\label{lem:G1reach}
Given an $\lts$ with a reachability acceptance condition, let $\lts'$ be $\lts$ where every almost final state of $\lts$ is made final in $\lts$. Then:
\begin{enumerate}
\item \label{item:inftofin} If Player $2$ wins $G_1(\lts)$ on infinite words, then Player $2$ also wins $G_1(\lts')$ on finite words;
\item \label{item:fintoinf} If $\lts'$ on finite words is history-deterministic, then so is $\lts$ on infinite words.
\end{enumerate} 
\end{lemma}

\begin{proof}
(\ref{item:inftofin}) Let $\strat$ be winning strategy for Player 2 in $G_1(\lts)$ on infinite words. Let the strategy $\strat'$ for Player 2 on $G_1(\lts')$ on finite words simply copy $\strat$.

If Player 2's run in a play of $G_1(\lts')$ that agrees with $\strat'$ reaches a state that was almost final in $\lts$, and therefore final in $\lts'$, it is winning for Player 2. If Player 2's run in a play of $G_1(\lts')$ that agrees with $\strat'$ does not reach such a state, we argue that Player 1's run also does not reach such a state. Indeed, towards a contradiction, assume that after some finite play $\pi$, Player 2's run has reached a state $q_2$ that is not final, while Player 1's run has reached a state $q_1$ that is final. Then, $\pi$ is also a play prefix that agrees with $\strat$ in $G_1(\lts)$. We argue that Player 1 can then win in $G_1(\lts)$, a contradiction: since $q_2$ is not almost final, he can play letters so that Player 2 does not build an accepting run; meanwhile, from $q_1$, he can use the strategy witnessing that it is almost final to build an accepting run on any word.

(\ref{item:fintoinf}) Let $\strat'$ be Player 2's winning strategy in the letter-game on $\lts'$ on finite words. In the letter-game on $\lts$ on infinite words, Player 2 follows $\strat'$ until the prefix of the word is in $L(\lts')$, at which point her run must have reached an almost final state of $\lts$, due to $\strat'$ being winning. From there, she can play the strategy witnessing that the state is almost final, and win.
If Player 1 never plays a prefix in $L(\lts')$ then the word he plays is not in $L(\lts)$, and Player 2 wins.
\end{proof}

\begin{theorem}
Given a fair LTS $\lts$ with a reachability acceptance condition, Player 2 wins $G_1(\lts)$ if and only if $\lts$ is history-deterministic. 
\end{theorem}

\begin{proof}
One direction is immediate since if Player 2 wins in the letter game, she can use the same strategy to win in $G_1(\lts)$ by ignoring Player 1's run.

For the other direction, if Player 2 wins $G_1(\lts)$, she also wins $G_1(\lts')$ on finite words, for $\lts'$ defined as in Lemma~\ref{lem:G1reach} above. Then, $\lts'$ is history-deterministic from Lemma~\ref{lem:G1}, and, again from Lemma~\ref{lem:G1reach}, $\lts$ is also history-deterministic.
\end{proof}

We now consider the problem of deciding whether a given safety or reachability TA
is history-deterministic. We use the observation that
the $k$-token games played on LTSs induced by TA can be expressed as a timed parity game from~\cite{CHP08} played on the $(k+1)$-fold product.

\begin{theorem}\label{thm:decidable}
Given a safety or reachability TA, deciding whether it is history-deterministic is decidable in \exptime.
\end{theorem}
\begin{proof}
From~\cref{lem:G1} and ~\cref{lem:G1reach}, deciding the history-determinism of a safety or reachability TA $\T$ reduces to solving $G_1(\T)$, which is a timed game played on the product of two copies of $\T$ (plus some intermediate states to encode the interaction in each round). The winning condition consists of a Boolean combination of safety or reachability conditions.

In the same way as done in the proof of \cref{lem:GFGTA:fairsim}, we can implement $G_1(\T)$ as a timed parity
game where the set states grows exponentially only in the number of colours in the parity condition.
The resulting timed parity game can be solved in time exponential in the number of clocks $c$ (see
\cref{lem:TPG}).
\end{proof}

As explained in the introduction, this also solves the good-enough synthesis problem of deterministic safety and reachability TA.

	\section{Synthesis}
	\label{sec_synthesis}
	\label{sec:synthesis}
	We show that as is the case in the regular~\cite{HP06}, pushdown~\cite{LZ22}, cost function~\cite{Col09}, and quantitative~\cite{BL21} settings, synthesis games with winning conditions given by history-deterministic TA are no harder to solve than those with winning condition given by deterministic TA.

\begin{definition}[Timed synthesis game]
	Given a timed language $L\subseteq (\Sigma_I\times \Sigma_O)^\omega_T$, the synthesis game for $L$ proceeds as follows.
	At turn $i$:
	\begin{itemize}
		\item Player~1 plays a delay $d_i$ and a letter $a_i\in \Sigma_I$
		\item Player~2 plays a letter $b_i\in \Sigma_O$.
	\end{itemize}
	Player~2 wins if $d_0\binom{a_0}{b_0}d_1\binom{a_1}{b_1}...\in L$ or if time does not progress.
If Player~2 has a winning strategy in the synthesis game, we say that $L$ is \emph{realisable}.
\end{definition}

\begin{restatable}{theorem}{synthesis}
    \label{thm:synthesis}
	Given a history-deterministic timed parity automaton $\T$, the synthesis game for $\lang(\T)$ is decidable and \exptime-complete.
\end{restatable}

The proof below follows a similar reduction to one in~\cite{LZ22}, in which the nondeterminism of the automaton is moved into Player~2's output alphabet, forcing her to simultaneously build a word in the winning condition and an accepting run witnessing this. Since accepting runs are recognized by deterministic automata, this reduces the problem to the synthesis problem for deterministic timed automata.	The lower bound follows from the \exptime-completeness of synthesis for deterministic TA~\cite{DM2002}.

The \exptime\ decidability of universality for history-deterministic TA follows both from the decidability of language inclusion in the previous section and from the decidability of synthesis: the universality of $\T$ reduces to deciding the winner of the synthesis game over $\{\binom{w}{w}\mid w\in \lang(\T)\}$, recognized by a history-deterministic TA if $\T$ is history-deterministic.

\begin{proof}[Proof of \cref{thm:synthesis}]
    For the upper bound, we reduce the problem to solving synthesis games for deterministic timed parity automata, which is in \exptime~\cite{DM2002}.
 
	Let $\T=(S,\iota,C,\Delta,\Sigma,\acc)$ be a timed automaton. Let $\T'$ be the deterministic timed automaton $(S,\iota,C,\Delta',\Sigma\times \Delta, \acc)$ where:

	$$\Delta' = \{(s,g,(\sigma,(s,g,\sigma,c,s')),c,s')| (s,g,\sigma,c,s')\in \Delta\}$$
	
	In other words, $\T'$ is a deterministic automaton with the state space of $\T$, over the alphabet $\Sigma\times \Delta$, where the transition in the input letter dictates the transition in the automaton. The language of $\T'$ is the set of words $(w,\rho)$ such that there is an accepting run of $\T$ over $w$ along the transitions of $\rho$.
	
	We now claim that given a history-deterministic automaton $\T$ with resolver $r$, Player~2 wins the synthesis game on $\T$ if and only if she wins it on $\T'$.
	First assume that Player~2 wins the synthesis game for $\T$ with a strategy $s$. Then, to win the synthesis game for $\T'$, at each turn $i$, after Player~1 plays $d_i$ and $a_i$, she needs to make two choices: she must choose both a response letter $b_i$ and a transition in $\T$ over $(a_i,b_i)$. Given Player~1's move and the (first component of the) word built so far, she can use the strategy $s$ to choose the response letter $b_i$; this guarantees that the first component of the play is a word accepted by $\T$. To choose the transition of $\T$, she can use the resolver $r$: given the run $\rho$ built from the delays (including $d_i$) and transitions played so far, she plays $r(\rho,(a_i,b_i))$. Since $r$ is a resolver, this strategy guarantees that the resulting run is accepting, and hence that she wins the synthesis game on $\T'$.
	
	On the other hand, if Player~1 wins the synthesis game on $\T$, he has a strategy $s$ which guarantees a play $w\in (\Sigma_i\times \Sigma_O)^T$ that is not in the language of $\T$. He can use the same strategy in the synthesis game of $\T'$ to guarantee a play $(w,\rho)$ such that $w$ is not in the language of $\T$, and by extension $(w,\rho)$ is not in the language of $\T'$, as there are no accepting runs over $w$ in $\T$.
	
	The lower bound follows from the \exptime-completeness of synthesis for deterministic TA~\cite{DM2002}.
\end{proof}

	\section{Conclusion}
	\label{sec_conclusion}
	We introduced history-determinism for timed automata and showed that several natural decision problems
-- timed language inclusion, universality and synthesis --
that are undecidable in general,
remain decidable for such automata.
We proved that for the important subclasses of timed safety and timed reachability automata, history-determinism can be checked (and therefore good-enough synthesis of deterministic reachability and safety automata can be solved) in exponential time.

We also showed that every history-deterministic timed safety or reachability automaton can be determinized, based on pruning the region automaton, and therefore these classes of automata are strictly less expressive than their fully non-deterministic counterparts.

We further establish that in terms of recognizing languages over infinite timed words, HD timed automata are strictly in between fully deterministic and the more general non-deterministic TA. More precisely, we present a history-deterministic one-clock coB\"uchi timed automata whose language is not recognized by any $k$-clock deterministic parity TA,
and a non-deterministic $1$-clock reachability TA whose language is not recognized by any HD $k$-clock parity TA.

We leave open the comparative expressiveness of history-deterministic B\"uchi timed automata:
are they determinizable? Already in the untimed case \cite{RK19},
positional resolvers for HD B\"uchi may not exist, but a finite amount of memory suffices.
We conjecture that HD B\"uchi timed automata can still be determinized via resolvers that
are based on regions and a finite amount of additional memory.

Let us conclude with another conjecture.
We showed that history-deterministic timed automata are ``good'' for solving
turn-based timed games, where in each turn of the game, one of the two players
chooses a time delay or an action. 
A more general, concurrent setting for timed games is presented in~\cite{DFHM03}.
In their concurrent version, both players simultaneously choose permissible pairs
of time delays and actions, and the player who has picked the shorter time delay
gets to move.
While concurrent games may not be determined, we conjecture that these concurrent
timed games can again be solved by composing the (timed) arena with the (timed)
winning condition, as long as the winning condition is history-deterministic.

	\bibliographystyle{alphaurl}
	\bibliography{bib/conferences,bib/journals,bib/references.clean}

\end{document}